\documentclass[table]{article} 
\usepackage[final]{neurips22} 

\input{header.tex}
\newcommand*\samethanks[1][\value{footnote}]{\footnotemark[#1]}
\newcommand*\samethanksagain[1][\value{footnote}]{\footnotemark[#1]}

\title{SALSA: Attacking Lattice Cryptography with Transformers}
\author{Emily Wenger\thanks{Equal contribution, work done while at Meta}  \\ University of Chicago \and {\bf Mingjie Chen}\samethanks \\ University of Birmingham \and {\bf Francois Charton}\thanks{Equal contribution, corresponding author: \texttt{fcharton@fb.com}} \\ Meta AI \and {\bf Kristin Lauter}\samethanksagain \\ Meta AI}

\begin{document}

\maketitle

\begin{abstract}
Currently deployed public-key cryptosystems will be vulnerable to attacks by full-scale quantum computers. Consequently, ``quantum resistant" cryptosystems are in high demand, and lattice-based cryptosystems, based on a hard problem known as Learning With Errors (LWE), have emerged as strong contenders for standardization. In this work, we train transformers to perform modular arithmetic and combine half-trained models with statistical cryptanalysis techniques to propose SALSA: a machine learning attack on LWE-based cryptographic schemes. SALSA can fully recover secrets for small-to-mid size LWE instances with sparse binary secrets, and may scale to attack real-world LWE-based cryptosystems.
\end{abstract}
    
\vspace{-0.3cm}
\section{Introduction}
\label{sec:intro}

\vspace{-0.3cm}

The looming threat of quantum computers has upended the field of cryptography. Public-key cryptographic systems have at their heart a difficult-to-solve math problem that guarantees their security. The security of most current systems (e.g.~\cite{rivest1978method,diffie1976new,miller1985use}) relies on problems such as integer factorization, or the discrete logarithm problem in an abelian group. Unfortunately, these problems are vulnerable to polynomial time quantum attacks on large-scale quantum computers due to Shor’s Algorithm~\cite{shor1994algorithms}. Therefore, the race is on to find new post-quantum cryptosystems (PQC) built upon alternative hard math problems. 

Several schemes selected for standardization in the 5-year NIST PQC competition are {\it lattice-based cryptosystems}, based on the hardness of the Shortest Vector Problem (SVP)~\citep{ajtai1996generating}, which involves finding short vectors in high dimensional lattices.  Many cryptosystems have been proposed based on hard problems which reduce to some version of the SVP, and known attacks are largely based on lattice-basis reduction algorithms which aim to find short vectors via algebraic techniques. 
The LLL algorithm~\cite{LLL} was the original template for lattice reduction, and although it runs in polynomial time (in the dimension of the lattice), it returns an exponentially bad approximation to the shortest vector. It is an active area of research~\cite{CN11_BKZ, MR09,ACJFP15} to fully understand the behavior and running time of a wide range of lattice-basis reduction algorithms, but the best known classical attacks on the PQC candidates run in time exponential in the dimension of the lattice. 

In this paper, we focus on one of the most widely used lattice-based hardness assumptions: {Learning With Errors} (LWE)~\cite{Reg05}. Given a dimension $n$, an integer modulus $q$, and a secret vector ${\bf s} \in \mathbb{Z}^n_q$, the {Learning With Errors problem} is to find the secret given noisy inner products $b := \textbf{a} \cdot \textbf{s}+e \mod q$, with ${\bf a} \in \mathbb{Z}^n_q$ a random vector, and $\bf e$ a small ``error'' sampled from a narrow centered Gaussian distribution (thus the reference to noise).
LWE-based encryption schemes encrypt a message by {\it blinding} it with noisy inner products.


The hardness assumption underlying Learning With Errors is that the addition of noise to the inner products makes the secret hard to discover.  In Machine Learning (ML), we often make the opposite assumption: given enough noisy data, we can still learn patterns from it. In this paper we investigate the possibility of training ML models to recover secrets from LWE samples.

To that end, we propose {\bf SALSA}, {a technique for performing {\bf S}ecret-recovery {\bf A}ttacks on {\bf L}WE via {\bf S}equence to sequence models with {\bf A}ttention}. SALSA trains a language model to predict $b$ from $\bf{a}$, and we develop two algorithms to recover the secret vector $\textbf{s}$ using this trained model.

Our paper has three main {\bf contributions}. We demonstrate that {\bf transformers can perform modular arithmetic} on integers and vectors. 
We show that transformers trained on LWE samples can be used to distinguish LWE instances from random data. 
This can be further turned into {\bf two algorithms that recover binary secrets}.
We show how these techniques yield a practical attack on LWE based cryptosystems and demonstrate its efficacy in the {\bf cryptanalysis of small and mid-size LWE instances with sparse binary secrets}.

\section{Lattice Cryptography and LWE}
\label{sec:lattice}

\begin{figure}[h] 
  \begin{center}
    \includegraphics[width=0.48\textwidth]{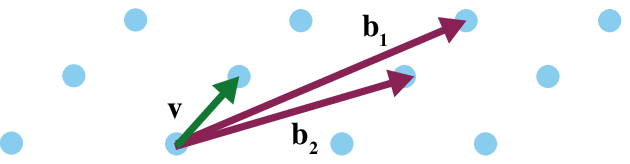}
  \end{center}
  \caption{\small The dots form a lattice $\Lambda$, generated by vectors $\textbf{b}_1,\,\textbf{b}_2$. $\bf v$ is the shortest vector in $\Lambda$.}
    \label{fig:lattice_vis}
\end{figure}

\subsection{Lattices and Hard Lattice Problems}\label{subsec:lattice}
An integer lattice of dimension $n$ over $\mathbb{Z}$ is the set of all integer linear combinations of $n$ linearly independent vectors in $\mathbb{Z}^n$. In other words, given $n$ such vectors $\textbf{v}_i \in \mathbb{Z}^n, i\in \mathbb{N}_n$, we define the lattice $\Lambda(\textbf{v}_1, .. \textbf{v}_n) := \{ \sum_{i=1}^{n}{a_i \textbf{v}_i} \mid a_i \in \mathbb{Z}\}.$ Given a lattice $\Lambda$, the Shortest Vector Problem (SVP) asks for a nonzero vector $\textbf{v} \in \Lambda$ with minimal norm. Figure~\ref{fig:lattice_vis} depicts a solution to this problem in the trivial case of a 2-dimensional lattice, where $\textbf{b}_1$ and $\textbf{b}_2$ generate a lattice $\Lambda$ and $\textbf{v}$ is the shortest vector in $\Lambda$. 

The best known algorithms to find exact solutions to SVP take exponential time and space with respect to $n$, the dimension of the lattice \cite{MV_SVP_exp}. There exist lattice reduction algorithms to find approximate shortest vectors, such as LLL \cite{LLL} (polynomial time, but exponentially bad approximation), or BKZ \cite{CN11_BKZ}. The shortest vector problem and its approximate variants are the hard mathematical problems that serve as the core of lattice-based cryptography. 
\subsection{LWE}\label{subsec:LWE}
\looseness=-1 The  Learning With Errors (LWE) problem, introduced in \cite{Reg05}, is parameterized by a dimension $n$, the number of samples $m$, a modulus $q$ and an error distribution $\chi$ (e.g., the discrete Gaussian distribution) over $\mathbb{Z}_q=\{0,\,1,\ldots,q-1\}$.
Regev showed that LWE is at least as hard as quantumly solving certain hard lattice problems. Later  \cite{Pei09a,LM09_hardness_lwe_exp,BLPRS13_poly_modulus_hardness}, showed LWE to be classically as hard as standard worst-case lattice problems, therefore establishing it as a solid foundation for cryptographic schemes.

\noindent\textbf{LWE and RLWE.}  The LWE distribution $\mathcal{A}_{\textbf{s},\chi}$ consists of pairs 
$(\bf A, \bf b) \it \in \ZZ_q^{m\times n}\times \ZZ_q^n$, where $\bf A$ is a uniformly random matrix in $\ZZ_q^{m\times n}$, $\bf b = \bf A \bf s + \bf e \text{ mod } \it q$, where $\bf s \it \in \ZZ_q^n$ is the secret vector sampled uniformly at random and $\bf e \in \it \ZZ_q^m$ is the error vector sampled from the error distribution $\chi$. We call the pair $(\textbf{A},\textbf{b})$ an LWE sample, yielding $n$ LWE instances: one row of $\textbf{A}$ together with the corresponding entry in $\bf{b}$ is one {\it LWE instance}. 
There is also a ring version of LWE, known as the Ring Learning with Errors (RLWE) problem (described further in Appendix~\ref{subsec:RLWE}).

\para{Search-LWE and Decision-LWE.} We now state the LWE hard problems. The search-LWE problem is to find the secret vector $\bf s$ given $(\bf A, \bf b)$ from $\mathcal{A}_{\textbf{s},\chi}$. The decision-LWE problem is to distinguish $\mathcal{A}_{\textbf{s},\chi}$ from the uniform distribution $\{(\bf A, \bf b) \in \it \ZZ_q^{m\times n}\times \ZZ_q^n$: $\bf A$ and $\bf b$ are chosen uniformly at random)$\}$. \cite{Reg05} provided a reduction from search-LWE to decision-LWE . We give a detailed proof of this reduction in Appendix~\ref{subsec:search_decision} for the case when the secret vector $\bf s$ is binary (i.e. its entries are 0 and 1). In Section~\ref{subsec:secret_guess}, our Distinguisher Secret Recovery method is built on this reduction proof.

\para{(Sparse) Binary secrets.} In LWE based schemes, the secret key vector $\textbf{s}$ can be sampled from various distributions. For efficiency reasons, binary distributions (sampling in $\{0, 1\}^{n}$) and ternary distributions (sampling in $\{-1, 0, 1\}^{n}$) are commonly used, especially in homomorphic encryption ~\cite{HES}. In fact, many implementations use a sparse secret with Hamming weight $h$ (the number of 1's in the binary secret). For instance, HEAAN uses $n = 2^{15},\,q = 2^{628}$, ternary secret and Hamming weight $64$~\cite{Cheon_hybrid_dual}. For more on the use of 
sparse binary secrets in LWE, see \cite{Albrecht2017_sparse_binary,Rachel_Player_sparse}.

\section{Modular Arithmetic with Transformers}
\label{sec:results_1D}
Two key factors make breaking LWE difficult: the presence of error and the use of modular arithmetic. Machine learning (ML) models tend to be robust to noise in their training data. In the absence of a modulus, recovering $\bf s$ from observations of $\bf a$ and $b = {\bf a} \cdot  {\bf s}+e$ merely requires linear regression, an easy task for ML.
Once a modulus is introduced, attacking LWE requires performing linear regression on an n-dimensional torus, a much harder problem. 

Modular arithmetic therefore appears to be a significant challenge for an ML-based attack on LWE. Previous research has concluded that modular arithmetic is difficult for ML models~\citep{palamasinvestigating}, and that transformers struggle with basic arithmetic~\citep{nogueira2021investigating}. However,~\cite{charton2021linear} showed that transformers can compute matrix-vector products, the basic operation in LWE, with high accuracy. As a first step towards attacking LWE, we investigate whether these results can be extended to the modular case. 

We begin with the one-dimensional case, training models to predict $b = a\cdot s \mod q$ from $a$, for some fixed unknown value of $s$, when $a,$ $s \in \mathbb{Z}_q$. This is a form of modular inversion since the model must implicitly learn the secret $s$ in order to predict the correct output $b$. We then investigate the $n$-dimensional case, with ${\bf a} \in \mathbb{Z}_q^n$ and $\bf s$ either in $\mathbb{Z}_q^n$ or in $\{0,1\}^n$ (binary secret). In the binary case, this becomes a (modular) subset sum problem.


\subsection{Methods}
\label{sec:1D_method}


{\bf Data Generation.} We generate training data by fixing the modulus $q$ (a prime with $15 \leq \ceil*{\log_2(q)} \leq 30$, see  the Appendix~\ref{sec:appx_addl_1D}), the dimension $n$, and the secret ${\bf s} \in \mathbb{Z}_q^n$ (or $\{0,1\}^n$ in the binary case). We then sample ${\bf a}$ uniformly in $\mathbb{Z}_q^n$ and compute $b={\bf a} \cdot {\bf s} \mod q$, to create data pair $({\bf a},b)$.

\looseness=-1 {\bf Encoding.} Integers are encoded in base B (usually, B=81), as a sequence of digits in $\{0, \dots B-1\}$. For instance, $(a,b) = (16,3)$ is represented as the sequences \texttt{[1,0,0,0,0]} and \texttt{[1,1]} in base $2$, or \texttt{[2,2]} and \texttt{[3]} in base $7$. In the multi-dimensional case, a special token separates the coordinates of $\bf{a}$. 

{\bf Model Training.} The model is trained to predict $b$ from $\bf a$, for an unknown but fixed value of $\bf s$. We use sequence-to-sequence transformers~\citep{transformer17} with one layer in the encoder and decoder, $512$ dimensions and $8$ attention heads. We minimize a cross-entropy loss, and use the Adam optimizer~\citep{kingma2014adam} with a learning rate of $5 \times 10^{-5}$. 
At epoch end ($300 000$ examples), model accuracy is evaluated over a test set of $10 000$ examples. We train until test accuracy is $95\%$ or loss plateaus for $60$ epochs.

\begin{figure*}[h]
\centering
    \begin{tabular}{cccccccccc} \toprule
    \multicolumn{1}{l}{\multirow{2}{*}{$\boldsymbol{\ceil{\log_2(q)}}$}} & \multicolumn{9}{c}{\textbf{Base}} \\
    \multicolumn{1}{l}{} & 2 & 3 & 5 & 7 & 24 & 27 & 30 & {\bf 81} & 128 \\ \midrule
    15 & 23 & 21 & 23 & 22 & 20 & 23 & 22 & \textbf{20} & 20 \\
    16 & 24 & 22 & 22 & 22 & 22 & 22 & 22 & {\bf 22} & 21 \\
    17 & - & 23 & 25 & 22 & 23 & 24 & 22 & {\bf 22} & 22 \\
    18 & - & 23 & 25 & 23 & 23 & 24 & 25 & {\bf 22} & 22 \\
    19 & - & 23 & - & 25 & 25 & 24 & - & {\bf 25} & 24 \\
    20 & - & - & - & - & 24 & 25 & 24 & {\bf 24} & 25 \\
    21 & - & 24 & - & 25 & - & - & - & - & 25 \\
    22 & - & - & - & - & - & 25 & - & - & 25 \\
    23 & - & - & - & - & - & - & - & - & - \\ \bottomrule
    \end{tabular}
    \captionof{table}{\small \textbf{Size of the training sets required for learning modular inversion.}  Base-2 logarithm of the number of examples needed to reach $95\%$ accuracy, for different values of $\ceil{\log_2(q)}$ and bases. '-' means $95\%$ accuracy not attained after $90$ million examples.}
    \label{tab:1Dspeed}
    \vspace{-0.3cm}
\end{figure*}

\begin{figure*}[h]
    \centering
    \includegraphics[width=0.7\linewidth]{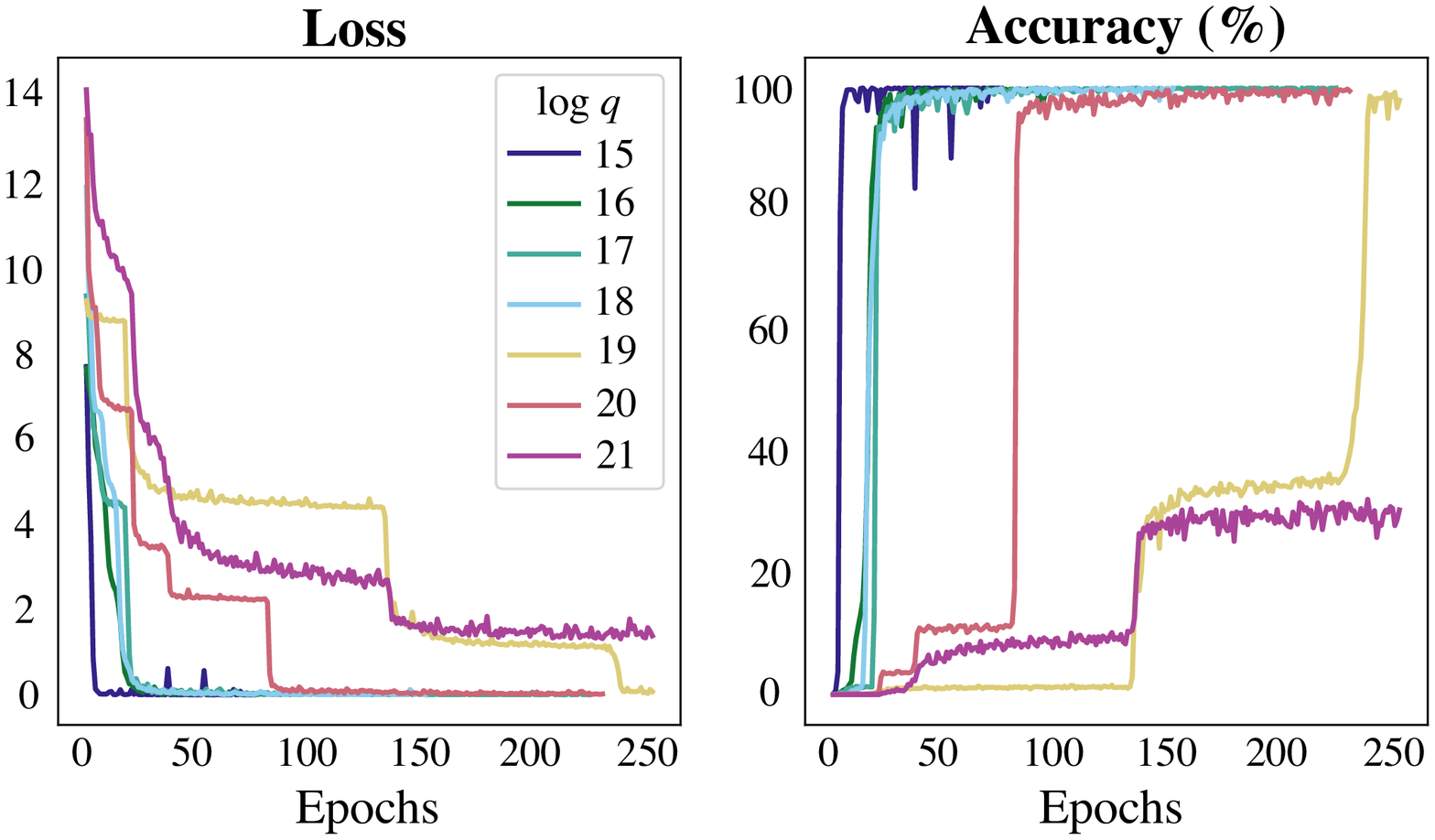}
    \caption{\small \textbf{Learning modular multiplication for various moduli}. {Test loss and accuracy for $q$ with $\ceil*{\log_2(q)}$ from 15 to 21. 300,000 training examples/epoch. One layer transformers with 512 dimensions, 8 attention heads, integers encoded in base 81.}}
    \vspace{-0.5cm}
\label{fig:curves}
\end{figure*}

\vspace{-0.25cm}
\subsection{Results}
\label{sec:modular_results}
\vspace{-0.25cm}

{\bf One-Dimensional.} For a fixed secret $s$, modular multiplication is a function from $\mathbb{Z}_q$ into itself, that can be learned by memorizing $q$ values. Our models learn modular multiplication with high accuracy for values of $q$ such that $\ceil*{\log_2(q)} \leq 22$. Figure~\ref{fig:curves} presents learning curves for different values of $log_2(q)$. The loss and accuracy curves have a characteristic step shape, observed in many of our experiments, which suggests that ``easier cases'' (small values of $\floor*{as/q}$) are learned first. 

The speed of learning and the training set size needed to reach high accuracy depend on the problem difficulty, i.e. the value of $q$. Table~\ref{tab:1Dspeed} presents the $\ceil{\log_2}$ of the number of examples needed to reach $95\%$ accuracy for different values of $\ceil{\log_2(q)}$ and base $B$. Since transformers learn from scratch, without prior knowledge of numbers and moduli, this procedure is not data-efficient. The number of examples needed to learn modular multiplication is between $10q$ and $50q$. Yet, these experiments prove that transformers can solve the modular inversion problem in prime fields. 

Table~\ref{tab:1Dspeed} illustrates an interesting point: learning difficulty depends on the base used to represent integers. For instance, base 2 and 5 allow the model to learn up to $\ceil*{\log_2(q)}=17$ and $18$, whereas base 3 and 7 can reach $\ceil*{\log_2(q)} = 21$. Larger bases, especially powers of small primes, enable faster learning. The relation between representation base and learning difficulty is difficult to explain from a number theoretic standpoint. Additional experiments are in Appendix~\ref{sec:appx_addl_1D}.

{\bf Multidimensional random integer secrets.} In the $n$-dimensional case, the model must learn the modular dot product between vectors $\bf a$ and $\bf s$ in $\ZZ_q^n$. The proves to be a much harder problem. For $n=2$, with the same settings, small values of $q$ ($251$, $367$ and $967$) can be learned with over $90\%$ accuracy, and $q=1471$ with $30\%$ accuracy. In larger dimension, all models fail to learn with parameters tried so far. Increasing model depth to 2 or 4 layers, or dimension to 1024 or 2048 and attention heads to 12 and 16, improves data efficiency (less training samples are needed), but does not scale to larger values of $q$ or $n>2$.

{\bf Multidimensional binary secrets.} Binary secrets make $n$-dimensional problems easier to learn. For $n=4$, our models solve problems with $\ceil{\log_2(q)}\leq 29$ with more than $99.5\%$ accuracy. For $n=6$ and $8$, we solve cases $\ceil{\log_2(q)} \leq 22$ with more than $85\%$ accuracy. But we did not achieve high accuracy for larger values of $n$. So in the next section, we introduce techniques for recovering secrets from a partially trained transformer. We then show that these additional techniques allow recovery of sparse binary secrets for LWE instances with $30 \le n \le 128$ (so far).

\vspace{-0.25cm}

\section{Introducing SALSA: LWE Cryptanalysis with Transformers}
\label{sec:attacks}

\vspace{-0.2cm}
Having established that transformers can perform integer modular arithmetic, we leverage this result to propose SALSA, a method for {\bf S}ecret-recovery {\bf A}ttacks on {\bf L}WE via {\bf S}eq2Seq models with {\bf A}ttention. 

\vspace{-0.25cm}
\subsection{SALSA Ingredients} 
\vspace{-0.25cm}

\looseness=-1 SALSA has three modules: a {\bf transformer model} $\mathcal{M}$, a {\bf secret recovery algorithm}, and a {\bf secret verification procedure}. We assume that SALSA has access to a number of LWE instances in dimension $n$ that use the same secret, i.e. pairs $({\bf a}, b)$ such that $b = {\bf a} \cdot {\bf s} + e \mod q$, with $e$ an error from a centered distribution with small standard deviation. SALSA runs in three steps. First, it uses LWE data to train $\mathcal{M}$ to predict $b$ given $\bf a$.
Next SALSA runs a secret recovery algorithm. It feeds $\mathcal{M}$ special values of $\bf a$, and uses the output $\tilde{b}=\mathcal{M}({\bf a})$ to predict the secret. 
Finally, SALSA evaluates the guesses $\tilde{\bf s}$ by verifying that residuals $r = b - {\bf a} \cdot \tilde{\bf s} \mod q$ computed from LWE samples have small standard deviation. If so, $\textbf{s}$ is recovered and SALSA stops. If not, SALSA returns to step 1, and iterates.

\vspace{-0.2cm}
\subsection{Model Training} 
\label{subsec:model_train}
\vspace{-0.2cm}
SALSA uses LWE instances to train a model that predicts $b$ from ${\bf a}$ by minimizing the cross-entropy between the model prediction $b'$ and $b$. The model architecture is a universal transformer \citep{dehghani2018universal}, in which a shared transformer layer is iterated several times (the output from one iteration is the input to the next). Our base model has two encoder layers, with 1024 dimensions and 32 attention heads, the second layer iterated 2 times, and two decoder layers with 512 dimensions and 8 heads, the second layer iterated 8 times. To limit computation in the shared layer, we use the copy-gate mechanism from \cite{csordas2021neural}. Models are trained using the Adam optimizer with $lr=10^{-5}$ and $8000$ warmup steps.

For inference, we use a beam search with depth $1$ (greedy decoding)~\citep{koehn2004pharaoh, sutskever2014sequence}.
At the end of each epoch, we compute model accuracy over a test set of LWE samples. Because of the error added when computing $b = {\bf a} \cdot {\bf s} +e$, exact prediction of $b$ is not possible. Therefore, we calculate {\it accuracy within tolerance $\tau$ ($acc_{\tau}$)}: the proportion of predictions $\tilde{b} =\mathcal{M}({\bf a})$ that fall within $\tau q$ of $b$, i.e. such that $\|b - \tilde{b}\| \le \tau q$. In practice we set $\tau=0.1$.

\vspace{-0.2cm}
\subsection{Secret Recovery}\label{subsec:secret_guess}
\vspace{-0.2cm}
We propose two algorithms for recovering $\bf s$: direct recovery from special values of $\bf a$, and distinguisher recovery using the binary search to decision reduction (Appendix~\ref{subsec:search_decision}). For theoretical justification of these, see Appendix~\ref{sec:appx_secret_guess}.

{\bf Direct Secret Recovery.} 

\begin{wrapfigure}{r}{0.5\textwidth}
    \begin{minipage}{0.5\textwidth}
    \vspace{-1.8cm}
    \begin{algorithm}[H]
      \caption{Direct Secret Recovery}
      \label{alg:direct_secret}
        \begin{algorithmic}[1]
          \STATE {\bfseries Input:} $\mathcal{M}, K, n$
          \STATE {\bfseries Output:} secret $s$
          \STATE $p = \mathbf{0}^n$
          \FOR{$i = 1, \dots, n$ }
          \STATE $a = \mathbf{0}^n; \; a_i = K$
          \STATE $p_i = \mathcal{M}(a)$
          \ENDFOR
          \STATE {\bfseries Return:} $s = binarize(p)$
        \end{algorithmic}
    \end{algorithm}
    \vspace{1.3cm}
    \end{minipage}
\end{wrapfigure}

The first technique, based on the LWE search problem, is analogous to a chosen plaintext attack.
For each index $i=1, \dots n$, a guess of the $i$-th coordinate of $\bf s$ is made by feeding model $\mathcal{M}$ the special value ${\bf a_i} = K {\bf e_i}$ (all coordinates of ${\bf e_i}$ are $0$ except the $i$-th), with $K$ a large integer. 
If $s_i = 0$, and the model $\mathcal{M}$ correctly approximates $b_i = {\bf a}_i\cdot {\bf s} + e$ from $\bf{a}_i$, then we expect $\tilde{b}_i := \mathcal{M}({\bf a}_i)$ to be a small integer; likewise if $s_i = 1$ we expect a large integer.   
This technique is formalized in  Algorithm~\ref{alg:direct_secret}. {The $binarize$ function in line $7$ is explained in Appendix~\ref{sec:appx_secret_guess}.} In SALSA, we run direct recovery with $10$ different $K$ values to yield $10$ $\bf s$ guesses.


{\bf Distinguisher Secret Recovery.}

\begin{wrapfigure}{r}{0.5\textwidth}
    \begin{minipage}{0.5\textwidth}
    \vspace{-4.35cm}
  \begin{algorithm}[H]
  \caption{Distinguisher Secret Recovery}
  \label{alg:distinguisher}
\begin{algorithmic}[1]
  \STATE {\bfseries Input:} $\mathcal{M}, n, q, acc_{\tau}, \tau$
  \STATE {\bfseries Output:} secret $s$
  \STATE $s = \mathbf{0}^n$
  \STATE $advantage, bound = acc_{\tau} - 2\cdot \tau,\tau \cdot q $
  \STATE $t = min\{50,\frac{2}{advantage^2}\}$
  \STATE $\mathbf{A_{LWE}}, \mathbf{B_{LWE}}$ = $LWESamples(t, n, q)$
  \FOR{$i = 1, \dots n$ }
    \STATE $\mathbf{A_{unif}} \sim \mathcal{U}\{0,q-1\}^{n \times t}$
    \STATE $\mathbf{B_{unif}} \sim \mathcal{U}\{0,q-1\}^{t}$
    \STATE $\mathbf{c} \sim \mathcal{U}\{0,q-1\}^t$ 
    \STATE $\mathbf{A'_{LWE}} = \mathbf{A_{LWE}}$
    \STATE $\mathbf{A'_{LWE}}[:,i] = (\mathbf{A_{LWE}}[:,i] + \mathbf{c}) \mod q$
    \STATE $\widetilde{\mathbf{B_{LWE}}} = \mathcal{M}(\mathbf{A'_{LWE}})$
    \STATE $\widetilde{\mathbf{B_{unif}}} = \mathcal{M}(\mathbf{A_{unif}})$
    \STATE $dl = |\widetilde{\mathbf{B_{LWE}}} - \mathbf{ B_{LWE}}|$
    \STATE $du = |\widetilde{\mathbf{B_{unif}}} - \mathbf{B_{unif}}|$
    \STATE $c_{LWE} = \#\{j \mid dl_j < bound , j\in \mathbb{N}_t\}$
    \STATE $c_{unif} = \#\{j \mid du_j < bound , j\in \mathbb{N}_t\}$
    \IF{($c_{LWE}-c_{unif}) \le advantage\cdot t/2$}
        \STATE $s_i = 1$
    \ENDIF
  \ENDFOR
  \STATE {\bfseries Return:} $s$
\end{algorithmic}
\end{algorithm}
\vspace{-1.2cm}
\end{minipage}
\end{wrapfigure}

The second algorithm for secret recovery is based on the decision-LWE problem. It uses the output of $\mathcal{M}$ to determine if LWE data $({\bf a},\,b)$ can be distinguished from randomly generated pairs $({\bf a_{r}}, b_{r})$. 
The algorithm for distinguisher-based secret recovery is shown in Algorithm~\ref{alg:distinguisher}.
At a high level, the algorithm works as follows. Suppose we have $t$ LWE instances $({\bf a},\,b)$ and $t$ random instances $({\bf a_{r}}, b_{r})$. For each secret coordinate $s_i$, we 
transform the $\bf a$ into $a'_i = a_i + c$, with $c \in \mathbb{Z}_q$ random integers. We then use model $\mathcal{M}$ to compute $\mathcal{M}({\bf a'})$ and $\mathcal{M}({\bf a_r})$. If the model has learned $\bf s$ and the $i^{th}$ bit of $\bf s$ is $0$, then $\mathcal{M}({\bf a'})$ should be significantly closer to $b$ than $\mathcal{M}({\bf a_r})$ is to $b_r$. Iterating on $i$ allows us to recover the secret bit by bit. SALSA runs the distinguisher recovery algorithm when model $acc_{\tau=0.1}$ is above $30\%$. This is the theoretical limit for this approach to work.


\subsection{Secret Verification.}\label{sec:check_secrets}

At the end of the recovery step, we have $10$ or $11$ guesses $\tilde{\bf s}$ (depending on whether the distinguisher recovery algorithm was run). To verify them, we compute the residuals $r = {\bf a} \cdot \tilde{\bf s} - b \mod q$ for a set of LWE samples $({\bf a}, b)$. If $\bf s$ is correctly guessed, we have $\tilde{\bf s}={\bf s}$, so $r = {\bf a} \cdot {\bf s} - b = e \mod q$ will be distributed as the error $e$, with small standard deviation $\sigma$. If $\tilde{\bf s} \ne {\bf s}$, $r$ will be (approximately) uniformly distributed over $\mathbb{Z}_q$ (because ${\bf a } \cdot \tilde {\bf s}$ and $b$ are uniformly distributed over $\mathbb{Z}_q$), and will have standard deviation close to $q/\sqrt{12}$.  Therefore, we can verify if $\tilde{\bf s}$ is correct by calculating the standard deviation of the residuals: if it is close to $\sigma$, the standard deviation of error, the secret was recovered.
In the case that $\sigma=3$ and $q=251$,  the standard deviation of $r$ is around $3$ if $\tilde{\bf s} = {\bf s}$, and $72.5$ if not.

\vspace{0.1cm}
\section{SALSA Evaluation}
\label{sec:lwe_results}

In this section, we present our experiments with SALSA. We generate datasets for LWE problems of different sizes, defined by the dimension and the sparsity of the binary secret. 
We use gated universal transformers, with two layers in the encoder and decoder. Default dimensions and attention heads in the encoder and decoder are 1024/512 and 16/4, but we vary them as we scale the problems. Models are trained on two NVIDIA Volta 32GB GPUs on an internal FAIR cluster. 

\subsection{Data generation}
We randomly generate LWE data for SALSA training/evaluation given the following parameters: dimension $n$, secret density $d$, modulus $q$, encoding base $B$, binary secret $s$, and error distribution $\chi$. For all experiments in this section, we use $q=251$ and $B=81$ (see \S\ref{sec:1D_method}), fix the error distribution $\chi$ to be a discrete Gaussian with $\mu=0, \sigma=3$ ~\cite{HES}, and randomly generate a binary secret $s$. 

We vary the problem size $n$ (the LWE dimension) and the density $d$ (the proportion of ones in the secret) to test our attack success and to observe how it scales.  For problem size, we experiment with $n=30$ to $n=128$. For density, we experiment with $0.002 \leq d \leq 0.15$. For a given $n$, we select $d$ so that the Hamming weight of the binary secret ($h=dn$), is larger than $2$. Appendix~\ref{sec:appx_2D_data} contains an ablation study of data parameters. 
We generate data using the RLWE variant of LWE, described in Appendix \ref{sec:appx_lwe}. For RLWE problems, each ${\bf a}$ is one line of a circulant matrix generated from an initial vector $\in \mathbb{Z}_q^n$. RLWE problems exhibit more structure than traditional LWE due to the use of the circulant matrix, which may help our models learn. 

\para{Note on RLWE parameter choices.} The choices of $n$, $q$, ring $R$, and error distribution determine the hardness of a given RLWE problem. In particular, prior work~\cite{eisentrager2014weak, elias2015provably, CLS1, CLS2} showed that many choices of polynomials defining the number ring $R$ are provably weak when $n$ is not a power of $2$. Thus, we first evaluate SALSA's success against RLWE for cyclotomic rings with dimension $n = 2^\ell$, $\ell \in \{5,6,7\}$ (see Table~\ref{tab:N_density_arch}). However, to help understand how SALSA scales with $n$, we also provide performance evaluations for $n$ values that are not powers of $2$, even though these RLWE settings may be subject to algebraic attacks more efficient than SALSA.

\begin{figure*}[h]
\centering
\begin{tabular}{cccc} \toprule
\textbf{\begin{tabular}[c]{@{}c@{}}Dim.\\ $n$\end{tabular}} & \textbf{\begin{tabular}[c]{@{}c@{}}Density\\ $d$\end{tabular}}  & \textbf{\begin{tabular}[c]{@{}c@{}}$log_2$ \\ samples\end{tabular}} & \textbf{\begin{tabular}[c]{@{}c@{}}Runtime \\ (hours)\end{tabular}} \\ \midrule
\multirow{2}{*}{30} & 0.1  & 20.93 & 1.2\\
 & 0.13 & 23.84 & 12.9 \\ \midrule
32 & 0.09 & 20.93 & 1.2 \\ \midrule
\multirow{2}{*}{50} & 0.06  & 22.25 & 4.7\\
 & 0.08 & 25.67 & 49.9 \\ \midrule
64 & 0.05 &  22.39 & 8 \\
70 & 0.04 & 22.74 & 11.9 \\ 
90 & 0.03 & 23.93 & 43.4 \\ 
110 & 0.03 & 24.07 & 68.8 \\ 
128 & 0.02 & 22.25 & 46.0 \\ \bottomrule
\end{tabular}
\captionof{table}{\small {\bf Full secret recovery.} {Highest density values at which the secret was recovered for each $n$, $q=251$. The model has 1024/512 dimension, 16/4 attention heads. (For $n = 128$, 1536/512 dimension, 32/4 attention heads).}}
\label{tab:N_density_arch}

\end{figure*}
\begin{figure*}[h]
    \centering
    \includegraphics[width=0.6\textwidth]{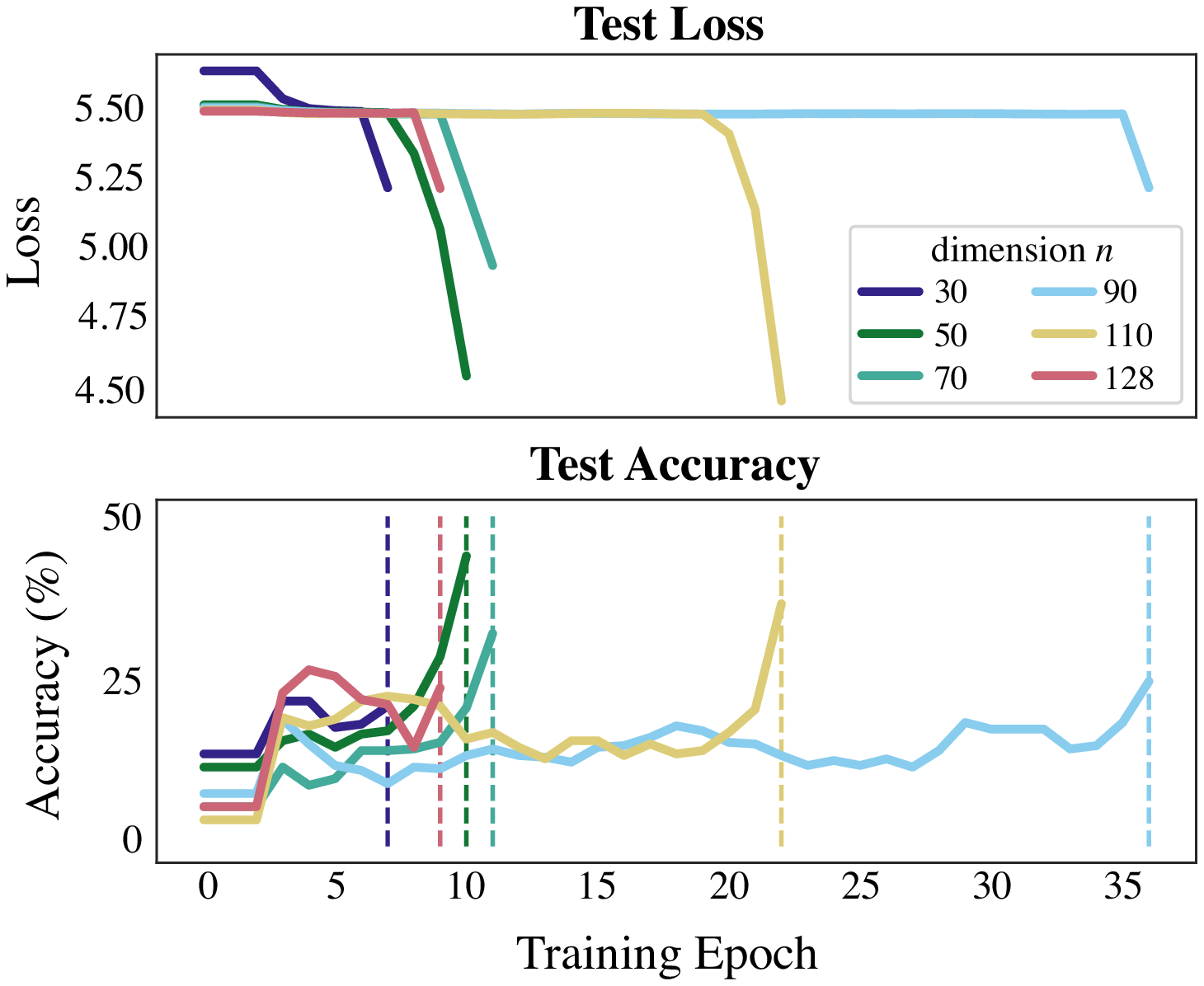}
    \vspace{-0.4cm}
    \captionof{figure}{\small {\bf Full secret recovery}: {Curves for loss and $acc_{\tau}=0.1$, for varying $n$ with Hamming weight $3$. For $n < 100$, model has 1024/512 embedding, 16/4 attention heads. For $n \ge 100$, model has 1536/512 embedding, 32/4 attention heads.}}
    \label{fig:fix_ham_vary_N}
\end{figure*}

\subsection{Results} 
\label{subsec:salsa_results}

Table~\ref{tab:N_density_arch} presents problem sizes $n$ and densities $d$ for which secrets can be fully recovered, together with the time and the logarithm of the number of training samples needed. SALSA can recover binary secrets with Hamming weight $3$ for dimensions up to $128$ ($2^6$). Secrets with Hamming weight 4 can be recovered for $n < 70$. 

For a fixed Hamming weight, the time needed to recover the secret increases with $n$, partly because the length of the input sequence fed into the model is proportional to $n$. On the other hand, the number of samples needed remains stable as $n$ grows. This observation is significant, because all the data used for training the model must be collected (e.g. via eavesdropping), making sample size an important metric. For a given $n$, scaling to higher densities requires more time and data, and could not be achieved with the architecture we use for $n>50$. As $n$ grows, larger models are needed: our standard architecture, with $1024/512$ dimensions and $16/4$ attention heads (encoder/decoder) was sufficient for $n \le 90$. For $n>90$, we needed $1536/512$ dimensions and $32/4$ attention heads. 

Figure~\ref{fig:fix_ham_vary_N} illustrates model behavior during training. After an initial burn-in period, the loss curve (top graph) plateaus until the model begins learning the secret. Once loss starts decreasing, model accuracy with $0.1q$ tolerance (bottom graph) increases sharply. Full secret recovery (vertical lines in the bottom graph) happens shortly after, often within one or two epochs. Direct secret recovery accounts for $55\%$ of recoveries, while the distinguisher only accounts for $18\%$  of recoveries (see Appendix~\ref{sec:sec_histogram}). $27\%$ of the time, both methods succeed simultaneously.

One key conclusion from these experiments is that the secret recovery algorithms enable secret recovery long before the transformer has been trained to high accuracy (even before training loss settles at a low level). Frequently, the model only needs {\it to begin} to learn for the attack to succeed.

\begin{table*}[t]
\centering
\begin{tabular}{cc?cc?ccc}\toprule
\multicolumn{2}{c?}{\begin{tabular}[c]{@{}c@{}}{\bf Regular vs. UTs}\\ (1024/512, 16/4, 8/8)\end{tabular}} &  \multicolumn{2}{c?}{\begin{tabular}[c]{@{}c@{}}{\bf Ungated vs. Gated}\\ (1024/512, 16/4, 8/8)\end{tabular}} & \multicolumn{3}{c}{\begin{tabular}[c]{@{}c@{}}{\bf UT Loops}\\ (1024/512, 16/4, {\bf X/X})\end{tabular}} \\ \midrule
 Regular & UT & Ungated & Gated & 2/8 & 4/4 & 8/2  \\ \midrule
 26.3 & 22.5 & 26.5 & 22.6 & 23.5 & 26.1 & 23.2   \\ \bottomrule
 \vspace{0.05cm}
\end{tabular}
\begin{tabular}{ccc?cccc}\toprule
 \multicolumn{3}{c?}{\begin{tabular}[c]{@{}c@{}}{\bf Encoder Dimension.}\\ ({\bf X}/512, 16/4, 2/8)\end{tabular}} & \multicolumn{4}{c}{\begin{tabular}[c]{@{}c@{}}{\bf Decoder Dimension}\\ (1024/{\bf X}, 16/4, 2/8)\end{tabular}} \\ \midrule
  512  & 2048 & 3040 & 256 & 768 & 1024 & 1536 \\ \midrule
   23.3 & 20.1 & 19.7 & 22.5 & 21.8 & 23.9 & 24.3 \\ \bottomrule
\end{tabular}

\caption{{\bf Architecture Experiments} We test the effect of model layers, loops, gating, and encoder dimension and report the $log_2$ samples required for secret recovery ($n=50$, Hamming weight $3$).}
\label{tab:layers_loops}

\end{table*}

\subsection{Experiments with model architecture}

SALSA's base model architecture is a Universal Transformer (UT) with a copy-gate mechanism.  Table~\ref{tab:layers_loops} demonstrates the importance of these choices. For problem dimension $n=50$, replacing the UT by a regular transformer with 8 encoder/decoder layers, or removing the copy-gate mechanism increases the data requirement by a factor of $14$. Reducing the number of iterations in the shared layers from 8 to 4 has a similar effect. Reducing the number of iterations in either the encoder or decoder (i.e. from 8/8 to 8/2 or 2/8) may further speed up training. Asymmetric transformers (e.g. large encoder and small decoder) have proved efficient for other math problems, e.g. \cite{Kasai2020}, \cite{charton2021linear}, and asymmetry helps SALSA as well. Table~\ref{tab:layers_loops} demonstrates that increasing the encoder dimension from 1024 to 3040, while keeping the decoder dimension at $512$, results in a $7$-fold reduction in sample size. Additional architecture experiments are presented in Appendix~\ref{sec:appx_2D_arch}.

\subsection{Effect of small batch size}

\begin{figure}[h]
    \centering
    \includegraphics[width=0.55\textwidth]{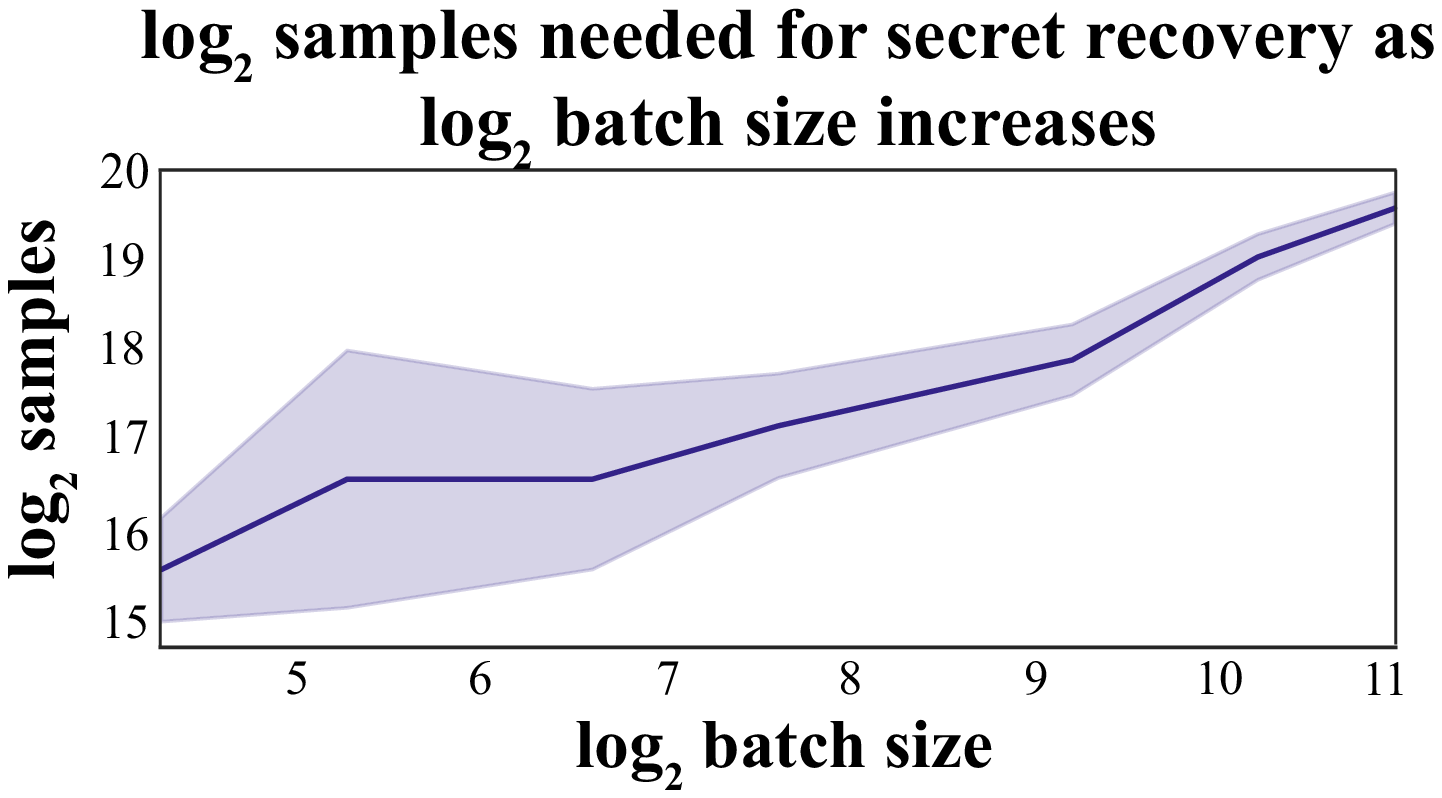}
    \vspace{-0.2cm}
    \caption{
    {\bf Batch size and sample efficiency.} Smaller batch sizes allow faster secret recovery.} 
    \label{fig:batch_size}
    \vspace{-0.2cm}
\end{figure}

We experiment with numerous training parameters to optimize SALSA's performance (e.g. optimizer, learning rate, floating point precision). While most of these settings do not substantively change SALSA's overall performance, we find that batch size has a significant impact on SALSA's sample efficiency. In experiments with $n=50$ and Hamming weight $3$, small batch sizes, e.g. $<50$, allow recovery of secrets with much fewer samples, as shown in Figure~\ref{fig:batch_size}. The same model architecture is used as for $n=50$ in Table~\ref{tab:N_density_arch}.

\subsection{Effect of varying modulus and base}
\label{sec:appx_2D_data}
Complex relationships between $n$, $q$, $B$, and $d$ affect SALSA's ability to fully recover secrets. Here, we explore these relationships, with success measured by the proportion of secret bits recovered. Table~\ref{tab:NvQ} shows SALSA's performance as $n$ and $q$ vary with fixed hamming weight $3$. SALSA performs better for smaller and larger values of $q$, but struggles on mid-size ones across all $N$ values (when hamming weight is held constant). These experiments on small dimension and varying $q$ can be directly compared to concrete outcomes of lattice reduction attacks on LWE for these sizes~\cite[Table 1]{CCLS}. Table~\ref{tab:hvQ} shows the interactions between $q$ and $d$ with fixed $n=50$. Here, we find that varying $q$ does not increase the density of secrets recovered by SALSA. Finally, Table~\ref{tab:BvsQ} shows the $\log_2$ samples needed for secret recovery with different input/output bases with $n=50$ and hamming weight $3$. The secret is recovered for all input/output base pairs except for $B_{in} = 17$, $B_{out}=3$, and using a higher input base reduces the $\log_2$ samples needed for recovery. 

\begin{table}[h]
\centering
\begin{tabular}{ccccccccccc}\toprule
\multirow{2}{*}{$n$} & \multicolumn{10}{c}{$\log_2(q)$} \\
 & 6 &  7 & 8 & 9 & 10 & 11 & 12 & 13 & 14 & 15 \\\midrule
{30} & \cellcolor{lred!80}  0.90 & \cellcolor{lightgreen!50} 1.0 & \cellcolor{lightgreen!50} 1.0 & \cellcolor{lightgreen!50} 1.0 & \cellcolor{lightgreen!50} 1.0 & \cellcolor{lightgreen!50} 1.0 & \cellcolor{lred!80} 0.9 & \cellcolor{lyellow!80} 0.97 & \cellcolor{lightgreen!50} 1.0 & \cellcolor{lightgreen!50} 1.0 \\

{50} &\cellcolor{lred!80}  0.94 & \cellcolor{lightgreen!50} 1.0 & \cellcolor{lightgreen!50} 1.0 & \cellcolor{lightgreen!50} 1.0 & \cellcolor{lightgreen!50} 1.0 & \cellcolor{lightgreen!50} 1.0 & \cellcolor{lred!80} 0.94 & \cellcolor{lyellow!80} 0.98 & \cellcolor{lightgreen!50} 1.0 & \cellcolor{lightgreen!50} 1.0 \\

{70} &\cellcolor{lred!80}  0.96 & \cellcolor{lightgreen!50} 1.0 & \cellcolor{lightgreen!50} 1.0 & \cellcolor{lightgreen!50} 1.0 & \cellcolor{lightgreen!50} 1.0 & \cellcolor{lightgreen!50} 1.0 & \cellcolor{lyellow!80} 0.96 & \cellcolor{lightgreen!50} 1.0 & \cellcolor{lightgreen!50} 1.0 & \cellcolor{lightgreen!50} 1.0 \\

{ 90} & \cellcolor{lred!80} 0.97 & \cellcolor{lyellow!80} 0.97 & \cellcolor{lightgreen!50} 1.0 & \cellcolor{lightgreen!50} 1.0 & \cellcolor{lyellow!80} 0.97 & \cellcolor{lightgreen!50} 1.0 & \cellcolor{lyellow!80} 0.97 & \cellcolor{lred!80} 0.97 & \cellcolor{lred!80} 0.97 & \cellcolor{lyellow!80} 0.99 \\\bottomrule
\end{tabular}

\vspace{0.1cm}
\caption{\small $\mathbf{n}$ vs $\mathbf{q}$. Results reported are proportion of total secret bits recovered for various $n$/$q$ combinations. {\color{lightgreen} Green cells} mean the secret was fully guessed,  {\color{lyellow} yellow cells} all the $1$ bits were correctly guessed during training, and {\color{lred} red cells} mean SALSA failed. Fixed parameters: $h=3$, $base_{in} = base_{out} = 81$. $1/1$ encoder layers, $1024/512$ embedding dimension, $16/4$ attention heads, $2/8$ loops.}
\vspace{-0.2cm}
\label{tab:NvQ}

\end{table}
\begin{table}[h]
\centering
\begin{tabular}{ccccccccccc}\toprule
\multirow{2}{*}{{d}} & \multicolumn{10}{c}{$\log_2(q)$} \\
 & 6 & 7 & 8 & 9 & 10 & 11 & 12 & 13 & 14 & 15 \\ \midrule

{0.06} & \cellcolor{lred!80}  0.94 & \cellcolor{lightgreen!50} 1.0 & \cellcolor{lightgreen!50} 1.0 & \cellcolor{lightgreen!50} 1.0 & \cellcolor{lightgreen!50} 1.0 & \cellcolor{lightgreen!50} 1.0 & \cellcolor{lred!80} 0.94 & \cellcolor{lyellow!80} 0.98 & \cellcolor{lightgreen!50} 1.0 & \cellcolor{lightgreen!50} 1.0 \\

{0.08} & \cellcolor{lred!80} 0.92 & \cellcolor{lred!80} 0.92 & \cellcolor{lightgreen!50} 1.0 & \cellcolor{lyellow!80} 0.92 & \cellcolor{lred!80} 0.94 & \cellcolor{lred!80} 0.92 &\cellcolor{lred!80} 0.94 & \cellcolor{lred!80} 0.94 &\cellcolor{lyellow!80} 0.94 & \cellcolor{lred!80} 0.94 \\

{0.10} & \cellcolor{lred!80} 0.90 &\cellcolor{lred!80} 0.94 & \cellcolor{lyellow!80} 0.96 & \cellcolor{lred!80} 0.90 & \cellcolor{lred!80} 0.90 & \cellcolor{lyellow!80} 0.92 & \cellcolor{lred!80} 0.90 & \cellcolor{lred!80} 0.92 & \cellcolor{lyellow!80} 0.94 & \cellcolor{lred!80} 0.92 \\ \bottomrule
\end{tabular}

\vspace{0.1cm}
\caption{\small $\mathbf{q}$ vs $\mathbf{d}$. Results reported are proportion of total secret bits recovered for various $q$/$d$ combinations. {\color{lightgreen} Green cells} mean the secret was fully guessed,  {\color{lyellow} yellow cells} all the $1$ bits were correctly guessed during training, and {\color{lred} red cells} mean SALSA failed. Parameters: $N=50$, $base_{in} = base_{out} = 81$. $1/1$ layers, $3040/1024$ embedding dimension, $16/4$ attention heads, $2/8$ loops.}
\label{tab:hvQ}
\vspace{-0.8cm}
\end{table}


\begin{table}[H]
\centering
\vspace{0.1cm} 
\begin{tabular}{c|ccccc} \toprule
\multirow{2}{*}{$B_{in}$} & \multicolumn{5}{c}{$B_{out}$} \\ 
 & 3 & 7 & 17 & 37 & 81 \\ \midrule
7 & 25.8 & \textbf{24.0} & 25.4 & 24.5 & 24.9 \\
17 & - & 25.9 & 27.2 & 25.6 & \textbf{25.4} \\
37 & 22.8 & \textbf{22.1} & 22.6 & 22.2 & 22.9 \\
81 & 22.2 & \textbf{22.1} & 22.4 & 21.9 & \textbf{22.1} \\ \bottomrule
\end{tabular}
\vspace{0.1cm}
\caption{\small {$\mathbf{B_{in}}$ v. $\mathbf{B_{out}}$.} Effect of input and output integer base representation on $\log_2$ samples needed for secret recovery. In each row, the \textbf{bold} numbers represent the lowest $\log_2$ samples needed for this value of $B_{in}$. Parameters: $n=50$, $h=3$, 2/2 layers, 1024/512 embedding dimension, 16/4 attention heads, and 2/8 loops.}
\label{tab:BvsQ}
\vspace{-0.8cm}
\end{table}

\subsection{Increasing dimension and density}
\label{sec:additional_results}

To attack real-world LWE problems, SALSA will have to successfully handle larger dimension $n$ and density $d$. Our experiments with architecture suggest that increasing model size, and especially encoder dimension, is the key factor to scaling $n$. Empirical observations indicate that scaling $d$ is a much harder problem. We hypothesize that this is due to the subset sum modular addition at the core of LWE with binary secrets. For a secret with Hamming weight $h$, the base operation $ {\bf a} \cdot {\bf s} + e \mod q$ is a sum of $h$ integers, followed by a modulus. For small values of $h$, the modulus operation is not always necessary, as the sum might not exceed $q$. As density increases, so does the number of times the sum ``wraps around'' the modulus, perhaps making larger Hamming weights more difficult to learn. To test this hypothesis, we limited the range of the coordinates in ${\bf a}$, so that $a_i < r$, with $r=\alpha q$ and $0.3 < \alpha < 0.7$. For $n=50$, we recovered secrets with density up to $0.3$, compared to $0.08$ with the full range of coordinates (see Table~\ref{tab:restrict_A}). Density larger than $0.3$ is no longer considered a sparse secret.

\begin{table}[h]
     \centering
    \begin{tabular}{cccccccc}\toprule
    \multirow{2}{*}{$d$} & \multicolumn{7}{c}{\bf Max $a$ value as fraction of $q$}                                   \\ 
                         &  $0.35$ & $0.4 $ & $0.45$ & $0.5 $ & $0.55$ & $0.6$ & $0.65 $ \\ \midrule
    $0.16$                    & \cellcolor{lightgreen!50}  1.0  & \cellcolor{lightgreen!50}  1.0  &  \cellcolor{lightgreen!50}  1.0     &  \cellcolor{lightgreen!50} 1.0    & \cellcolor{lightgreen!50} 1.0 & \cellcolor{lightgreen!50} 1.0    & \cellcolor{lyellow!80} 0.88      \\
    
    $0.18$                    &   \cellcolor{lightgreen!50}  1.0   &  \cellcolor{lightgreen!50}  1.0  & \cellcolor{lightgreen!50}  1.0   & \cellcolor{lightgreen!50} 1.0  &     \cellcolor{lred!80} 0.82    &   \cellcolor{lred!80}   0.86  &  \cellcolor{lred!80}   0.84 \\

    $0.20$                  &     \cellcolor{lightgreen!50}  1.0    &  \cellcolor{lightgreen!50}  1.0     &   \cellcolor{lightgreen!50}  1.0    & \cellcolor{lightgreen!50} 1.0  &     \cellcolor{lightgreen!50} 1.0   & \cellcolor{lred!80} 0.82  &  \cellcolor{lred!80} 0.82   \\
    
    $0.22$                   &    \cellcolor{lyellow!80} 0.98  &    \cellcolor{lightgreen!50}  1.0   &  \cellcolor{lightgreen!50}  1.0    &     \cellcolor{lyellow!80} 0.98  &      \cellcolor{lred!80} 0.80   &    \cellcolor{lred!80}  0.78  &  \cellcolor{lyellow!80}  0.86  \\
    
    $0.24$                  & \cellcolor{lightgreen!50}  1.0   &   \cellcolor{lightgreen!50}  1.0   & \cellcolor{lightgreen!50} 1.0   &     \cellcolor{lyellow!80} 0.98   &    \cellcolor{lred!80} 0.78    &    \cellcolor{lred!80} 0.78   &  \cellcolor{lyellow!80} 0.80 \\
    
    $0.26$                   & \cellcolor{lightgreen!50} 1.0  & \cellcolor{lightgreen!50} 1.0   &    \cellcolor{lyellow!80} 0.88     &        \cellcolor{lyellow!80} 0.92  &   \cellcolor{lyellow!80}   0.76   &  \cellcolor{lred!80}  0.76    &  \cellcolor{lred!80}  0.76    \\
    
    $0.28$                   &    \cellcolor{lyellow!80} 0.98     & \cellcolor{lightgreen!50}  1.0  &       \cellcolor{lyellow!80} 0.80   &         \cellcolor{lred!80} 0.74 &   \cellcolor{lred!80}   0.74   &    \cellcolor{lred!80} 0.76   &  \cellcolor{lred!80}  0.74    \\
    
    $0.30$                   &    \cellcolor{lyellow!80} 0.98      &   \cellcolor{lightgreen!50}  1.0 &     \cellcolor{lred!80} 0.93    &   \cellcolor{lred!80} 0.76    &   \cellcolor{lred!80}  0.72    &    \cellcolor{lred!80} 0.74   &  \cellcolor{lred!80} 0.74  \\ \bottomrule
    \end{tabular}
    \vspace{0.1cm}
        \captionof{table}{
        {\bf Secret recovery when max $\bf a$ value is bounded.} Results shown are fraction of the secret recovered by SALSA for $n=50$ with varying $d$ when $a$ values are $\le p \cdot Q$. {\color{lightgreen} Green} means that $s$ was fully recovered. {\color{lyellow} Yellow} means all $1$ bits were recovered, but not all $0$ bits. {\color{lred} Red} means SALSA failed. }
    \label{tab:restrict_A}
    \vspace{-0.2cm}
\end{table}

\vspace{-0.3cm}
\subsection{Increasing error size}
\label{subsec:error_size}
\vspace{-0.3cm}

\begin{figure*}[h]
\begin{minipage}[c]{0.480\textwidth}
    \centering
    \includegraphics[width=1.05\textwidth]{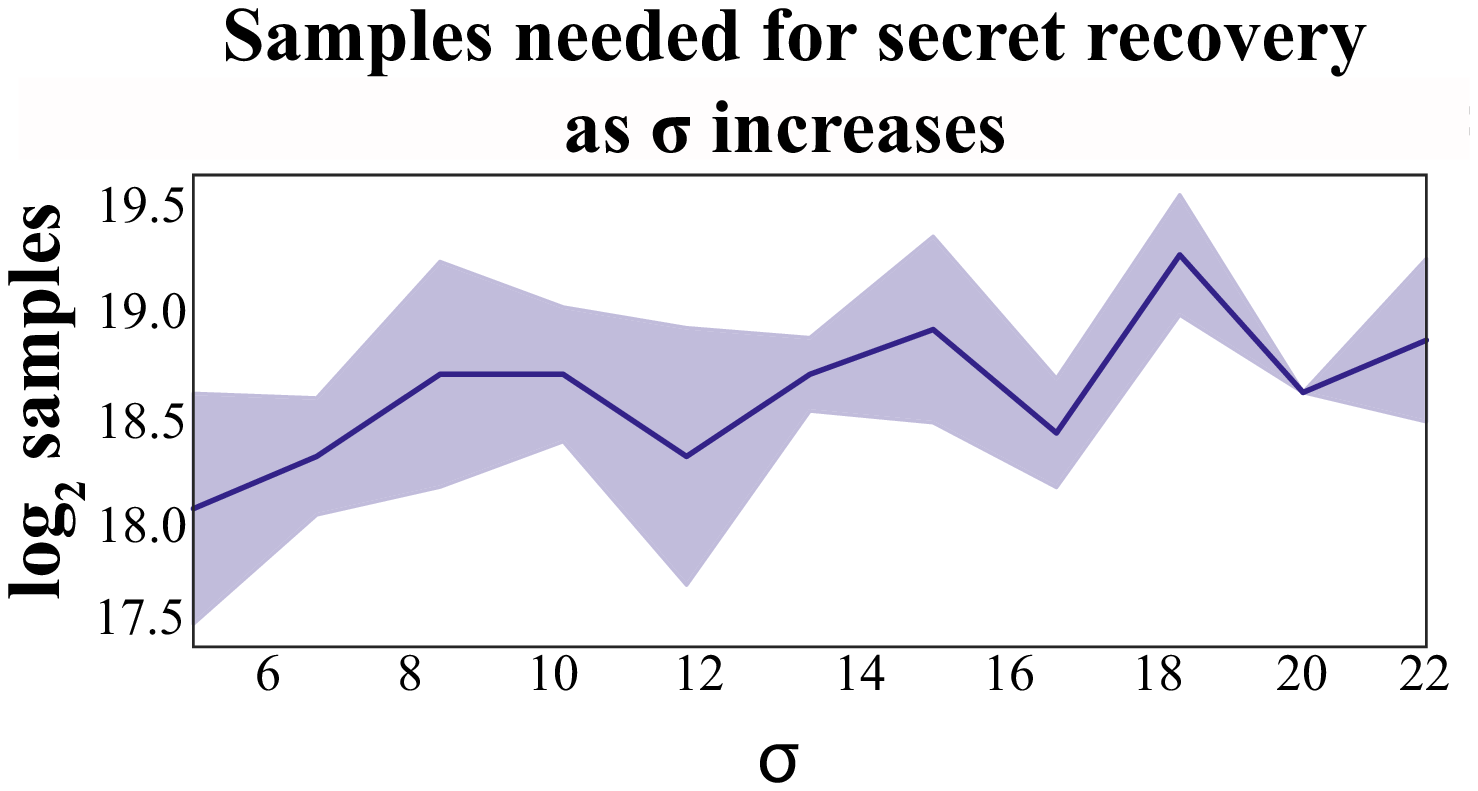}
    \vspace{-0.6cm}
    \caption{\small {\bf $log_2$ samples vs. $\sigma$.} As $\sigma$ increases, $log_2$ samples required for $n=50$, $h=3$ increases linearly.}
    \label{fig:sigma_fixed_N}
    \end{minipage}\hfill
\begin{minipage}[c]{0.48\textwidth}    
    \centering
    \vspace{0.7cm}
\begin{tabular}{ccccc}
    \toprule
    {n / $\sigma$} & 30/5 & 50/7 & 70/8 & 90/9 \\\midrule
    {log Samples} & 18.0 & 18.5 & 19.3 & 19.6 \\ \bottomrule
    \end{tabular}
    \vspace{1.1cm}
    \captionof{table}{
    \small
    \textbf{$log_2$ samples needed for secret recovery when $\sigma = \lfloor \sqrt{n} \rfloor$.} Results averaged over $6$ SALSA runs at each $n/\sigma$ level.}
    \label{tab:sigma}
    \end{minipage}
\end{figure*}

Theoretically for lattice problems to be hard, $\sigma$ should scale with $\sqrt{n}$, although this is often ignored in practice, e.g.~\cite{HES}. Consequently, we run most SALSA experiments with $\sigma=3$, a common choice in existing RLWE-based systems.  Here, we investigate how SALSA performs as $\sigma$ increases. First, to match the theory, we run experiments where $\sigma = \lfloor \sqrt{n} \rfloor$, $h=3$ and found that SALSA recovers secrets even as $\sigma$ scales with $\sqrt{n}$ (see Table~\ref{tab:sigma}, same model architecture as Table~\ref{tab:N_density_arch}). Second, we evaluate SALSA's performance for fixed $n$ and $h$ values as $\sigma$ increases. We fix $n=50$ and $h=3$ and evaluate for $\sigma$ values up to $\sigma=24$. Secret recovery succeeds for all tests, although the number of samples required for recovery increases linearly (see Figure~\ref{fig:sigma_fixed_N} in Appendix). For both sets of experiments, we reuse samples up to $10$ times. 

\section{SALSA in the Wild}
\label{sec:discussion}

\subsection{Problem Size}
\label{subsec:problem_size}
Currently, SALSA can recover secrets from LWE samples with $n$ up to $128$ and density $d=0.02$. It can recover higher density secrets for smaller $n$ ($d=0.08$ when $n=50$). 
Sparse binary secrets are used in real-world RLWE-based homomorphic encryption implementations, and attacking these is a future goal for SALSA. To succeed, SALSA will need to  scale to attack larger $n$. 
Other parameters for full-strength homomorphic encryption such as secret density, are within SALSA's current reach, (the secret vector in HEAAN has $d < 0.002$) and error size (~\cite{HES} recommends $\sigma=3.2$). 

Other LWE-based schemes use dimensions that seem achievable given our current results. For example, in the LWE-based public key encryption scheme Crystal-Kyber~\cite{crys_kyber}, the secret dimension is $k\times 256$ for $k = \{2,\,3,4\}$, an approachable range for SALSA. The LWE-based signature scheme Crystal-Dilithium has similar sizes for $n$~\cite{crys_dilithium}. However, these schemes don't use sparse binary secrets, and adapting SALSA to non-binary secrets is a non-trivial avenue for future work.

\begin{figure*}[t]
    \centering
    \includegraphics[width=0.6\textwidth]{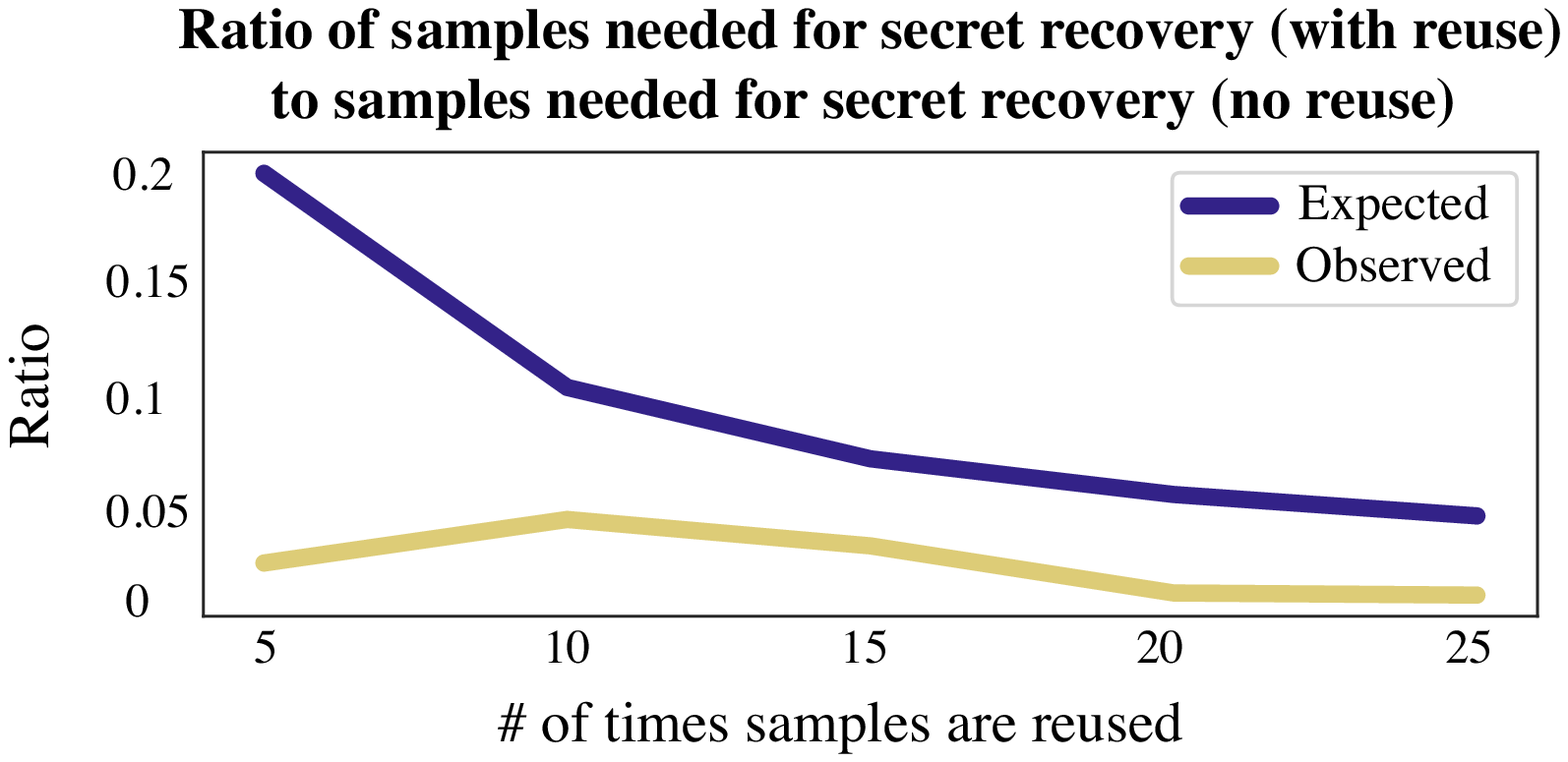}
    \caption{\small {\bf Reusing LWE samples yields a significant decrease in the number of samples needed for secret recovery.} Shown here is the ratio of samples required for secret recovery with reuse to the samples required for secret recovery without reuse, both expected (top curve) and observed (bottom curve, better than expected). }
    \label{fig:sample_efficency}
\end{figure*}
\begin{table*}[t]

    \centering
\vspace{0.1cm}
\centering
\begin{tabular}{c|ccccc}\toprule
\multirow{2}{*}{K} & \multicolumn{5}{c}{\textbf{Times Samples Reused}} \\
 & 5 & 10 & 15 & 20 & 25 \\ \midrule
1 & 20.42 & 21.915 & 20.215 & 17.610 & 17.880 \\
2 & 19.11 & 20.605 & 18.695 & 18.650 & 16.490 \\
3 & 20.72 & 19.825 & 17.395 & 18.325 & 16.200 \\
4 & 19.11 & 19.065 & 17.180 & 15.405 & 16.355 \\ \bottomrule
\end{tabular}
\caption{\small {\bf Effect of Sample Reuse on Sample Efficiency.} Sample reuse via linear combinations greatly improves sample efficiency. The secret is recovered in all experiments, indicating that error introduced by sample combination does not degrade performance. Parameters: $n=50$, Hamming 3, 2/2 encoder layers, 1024/512 embedding, 16/4 attention heads, 2/8 loops.}
\label{tab:appx_sample_reuse}

\end{table*}

\subsection{Sample Efficiency}
\looseness=-1 A key requirement of real-world LWE attacks is sample efficiency. In practice, an attacker will only have access to a small set of LWE instances $({\bf a}, b)$ for a given secret $s$. For instance, in Crystal-Kyber, there are only $(k+1)n$ LWE instances available with $k = 2,\,3$ or $4$ and $n = 256$.
The experiments in \cite{CCLS,bai_galbraith} use fewer than 500 LWE instances. The TU Darmstadt challenge provides $n^2$ LWE instances to attackers.

The $\log_2 samples$ column of Table~\ref{tab:N_density_arch} lists the number of LWE instances needed for model training. This number is much larger than what is likely available in practice, so it is important to reduce sample requirements. Classical algebraic attacks on LWE require LWE instances to be linearly independent, but {\it SALSA does not have this limitation.} Thus, we can reduce SALSA's sample use in several ways. 
First, we can reuse samples during training. Figure~\ref{fig:sample_efficency} confirms that this allows secret recovery with fewer samples. Second, we can use integer linear combinations of given LWE samples to make new samples which have the same secret but a larger error $\sigma$, explained below. Using this method, we can generate up to $2^{42}$ new samples from 100 original samples. 

\para{Generating New Samples.} It is possible to generate new LWE samples from existing ones via linear combinations. Assume we have access to $m$ LWE samples, and suppose a SALSA model can still learn from samples from a family of Gaussian distributions with standard deviations less than $N\sigma$, where $\sigma$ is the standard deviation of the original LWE error distribution. Experimental results show that SALSA's models are robust to $\sigma$ up to $24$ (see \S\ref{subsec:error_size}). With these assumptions, the number of new samples we could make is equal to the number of vectors $\textbf{v} = (v_1,\ldots,v_m)^T \in \ZZ^m$ such that $\sum_{i=1}^m |v_i| \leq N^2.$ For simplicity, assume that $v_i's$ are nonnegative. Then,  there are $\sum_{n=1}^{N^2} \frac{(m+n-1)!}{(m-1)!(n)!}$ new LWE samples one can generate.

\para{Results on Generated Samples.} Next, we show how SALSA performs when we combine different numbers of existing samples to create new ones for model training. We use the above method but do not allow the same sample to appear more than once in a given combination. We fix $K$, which is the number of samples used in each linear combination of reused samples. Then, we generate $K$ coefficients for the combined samples, where each $k_i$ is randomly chosen from $\{-1,0,1\}$. Finally, we randomly select $K$ samples from a pre-generated set of samples, and produce a new sample from their linear combination with the $k_i$ coefficients. These new samples follow error distribution with the standard deviation less than or equal to $\sqrt{K}\sigma$.

\looseness=-1 We experiment with different values of $K$, as well as different numbers of times we reuse a sample in linear combinations before discarding it. The $\log_2$ samples required for secret recovery for each ($K$, times reused) setting are reported in Table~\ref{tab:appx_sample_reuse}. The first key result is that the secret is recovered in all experiments, confirming that the additional error introduced via sample combination does not disrupt model learning. Second, as expected, sample requirements decrease as we increase $K$ and times reused.

\subsection{Comparison to Baselines}
Most existing attacks on LWE such as uSVP and dual attack use an algebraic approach that involves building a lattice from LWE instances such that this lattice contains an exceptionally short vector which encodes the secret vector information. Attacking LWE then involves finding the short vector via lattice reduction algorithms like BKZ \cite{CN11_BKZ}.  For LWE with sparse binary secrets, the main focus of this paper, various techniques can be adapted to make algebraic attacks more efficient. \cite{CCLS, bai_galbraith} and~\cite{Cheon_hybrid_dual} provide helpful overviews of algebraic attacks on sparse binary secrets. More information about attacks on LWE is in Section ~\ref{subsec:lwe_attack}.

Compared to existing attacks, SALSA's most notable feature is its novelty. We do not claim to have better runtime, neither do we claim the ability to attack real-world LWE problems (yet). Rather, we introduce a new attack and demonstrate with non-toy successes that transformers can be used to attack LWE.  Given our goal, no serious SALSA speedup attempts have been made so far, but a few simple improvements could reduce runtime. First, the slowest step in SALSA is model training, which can be greatly accelerated by distributing it across many GPUs. 
Second, our transformers are trained from scratch, so pre-training them on such basic tasks as modular arithmetic could save time and data. Finally, the amount of training needed before the secret is recovered depends in large part on the secret guessing algorithms.
New algorithms might allow SALSA to recover secrets faster.

Since SALSA does not involve finding the shortest vector in a lattice, it has an advantage over the algebraic attacks -- with all LWE parameters fixed and in the range of SALSA, SALSA can attack the LWE problem for a smaller modulus $q$ compared to the algebraic attacks. This is because the target vector is relatively large in the lattice when $q$ is smaller and is harder to find. For instance, in \cite{CCLS}, their Table 2 shows that when the block size is $45$, for $n = 90$, their attack does not work for $q$ less than 10 bits, but we can handle $q$ as small as 8 bits (Table~\ref{tab:NvQ}).

\subsection{Overview of Attacks on LWE}\label{subsec:lwe_attack}

\looseness=-1 Typically, attacks on the LWE problem use an algebraic approach and involve lattice reduction algorithms such as BKZ \cite{CN11_BKZ}. The LWE problem can be turned into a BDD problem (Bounded Distance Decoding) by considering the lattice generated by LWE instances, and BDD can be solved by Babai's Nearest Plane algorithm \cite{LP11} or pruned enumeration \cite{LN13}, this is known as the primal BDD attack. The primal uSVP attack constructs a lattice via Kannan's embedding technique \cite{Kan87} whose unique shortest vector encodes the secret information. The Dual attack \cite{MR09} finds a short vector in the dual lattice which can be used to distinguish the LWE samples from random samples. Moreover, there are also attacks that do not use lattice reduction. For instance, the BKW style attack \cite{ACJFP15} uses combinatorial methods; however, this assumes access to an unbounded number of LWE samples.

Binary and ternary secret distributions are widely used in homomorphic encryption schemes. In fact, many implementations even use a sparse secret with Hamming weight $h$. In \cite{BLPRS13_poly_modulus_hardness} and \cite{Micc18}, both papers give reductions of binary-LWE to hard lattice problems, implying the hardness of binary-LWE. Specifically, the $(n,q)$-binary-LWE problem is related to a $(n/t,q)$-LWE problem where $t = O(\log(q))$. For example, if $n = 256$ is a hard case for uniform secret, we can be confident that binary-LWE is hard for $n = 256\log(256) = 2048$. But \cite{bai_galbraith} refines this analysis and gives an attack against binary-LWE. Their experimental results suggest that increasing the secret dimension by a $\log(\log(n))$ factor might be already enough to achieve the same security level for the corresponding LWE problem with uniform secrets.

Let us now turn to the attacks on (sparse) binary/ternary secrets. The uSVP attack is adapted to binary/ternary secrets in \cite{bai_galbraith}, where a balanced version of Kannan's embedding is considered. This new embedding increases the volume of the lattice and hence the chance that lattice reduction algorithms will return the shortest vector. The Dual attack for small secret is considered in \cite{Albrecht2017_sparse_binary} where the BKW-style techniques are combined. The BKW algorithm itself also has a binary/ternary-LWE variant \cite{AFFP14_BKW_small}. Moreover, several
additional attacks are known which can exploit the sparsity of an LWE secret, such as \cite{JFRT16,Cheon_hybrid_dual} . All of these techniques use a combinatorial search in some dimension $d$, and then follow by solving a lattice problem in dimension $n-d$. For sparse secrets, this is usually more efficient than solving the original lattice problem in dimension $n$. 

\section{Related Work}
\label{sec:back}

{\bf Use of ML for cryptanalysis.} The fields of cryptanalysis and machine learning are closely related~\cite{rivest1991cryptography}. Both seek to approximate an unknown function $\mathcal{F}$ using data, although the context and techniques for doing so vary significantly between the fields. Because of the similarity between the domains, numerous proposals have tried to leverage ML for cryptanalysis. ML-based attacks have been proposed against a number of cryptographic schemes, including block ciphers~\cite{alani2012neuro, so2020deep, kimura2021output, baek2020recent, gohr2019improving, benamira2021deeper, chen2021bridging}, hash functions~\cite{goncharov2019using}, and substitution ciphers~\cite{ahmadzadeh2021novel, srivastava2018learning, aldarrab2020can}. Although our work is the first to use recurrent neural networks for lattice cryptanalysis, prior work has used them for other cryptographic tasks. For example, \cite{greydanus2017learning} showed that LSTMs can learn the decryption function for polyalphabetic ciphers like Enigma. Follow-up works used variants of LSTMs, including transformers, to successfully attack other substitution ciphers~\cite{ahmadzadeh2021novel, srivastava2018learning, aldarrab2020can}.

{\bf Use of transformers for mathematics.} 
The use of language models to solve problems of mathematics has received much attention in recent years. A first line of research explores math problems set up in natural language. \cite{saxton2019analysing} investigated their relative difficulty, using LSTM \citep{hochreiter1997long} and transformers, 
while \cite{griffith021} showed large transformers could achieve high accuracy on elementary/high school problems.
A second line explores various applications of transformers on formalized symbolic problems. \cite{lample2019deep} showed that symbolic math problems 
could be solved to state-of-the-art accuracy with transformers. \cite{welleck2021symbolic} discussed their limits when generalizing out of their training distribution. Transformers have been applied to dynamical systems \citep{charton2020learning}, transport graphs \citep{charton2021deep}, theorem proving \citep{polu2020generative}, SAT solving \citep{shi2021transformerbased}, and symbolic regression \citep{biggio2021neural, dascoli2022deep}.  
A third line of research focuses on arithmetic/numerical computations and has had slower progress. \cite{palamasinvestigating} and \cite{nogueira2021investigating} discussed the difficulty of performing arithmetic operations with language models. Bespoke network architectures have been proposed for arithmetic operations~\citep{kaiser2015neural, trask2018neural}, and transformers were used for addition and similar operations \citep{power2022grokking}. \cite{charton2021linear} showed that transformers can learn  numerical computations, such as linear algebra, and introduced the shallow models with shared layers used in this paper. 

\section{Conclusion}
In this paper, we demonstrate that transformers can be trained to perform modular arithmetic. Building on this capability, we design SALSA, a method for attacking the LWE problem with binary secrets, a hardness assumption at the foundation of many lattice-based cryptosystems. 
We show that SALSA can break LWE problems of medium dimension (up to $n=128$), comparable to those in the Darmstadt challenge \cite{darmstadt_lwe}, with sparse binary secrets. 
This is the first paper {to} use transformers to solve hard problems in lattice-based cryptography. Future work will attempt to scale up SALSA to attack higher dimensional lattices with more general secret distributions.

The key to scaling up to larger lattice dimensions seems to be to increase the model size, especially the dimensions, the number of attention heads, and possibly the depth. 
Large architectures should scale to higher dimensional lattices such as $n=256$ which is used in practice. Density, on the other hand, is constrained by the performance of transformers on modular arithmetic. Better representations of finite fields could improve transformer performance on these tasks. Finally, our secret guessing algorithms enable SALSA to recover secrets from low-accuracy transformers, therefore reducing the data and time needed for the attack. Extending these algorithms to take advantage of partial learning should result in better performance.

\newpage
\bibliographystyle{apalike}
\bibliography{main}

\newpage
\appendix
\begin{center} \Large \bf Appendix \end{center}

\section{Further Details of LWE}
\label{sec:appx_lwe}

\subsection{Ring Learning with Errors (\S\ref{sec:lattice})}\label{subsec:RLWE}

We now define RLWE samples and explain how to get LWE instances from them. Let $n$ be a power of 2, and let $R_q = \ZZ_q[x]/(x^n+1)$ be the set of polynomials whose degrees are at most $n-1$ and coefficients are from $\ZZ_q$. The set $R_q$ forms a ring with additions and multiplications defined as the usual polynomial additions and multiplications in $\ZZ_p[x]$ modulo $x^n+1$. One RLWE sample refers to the pair \[(a(x),b(x) : = a(x)\cdot s(x) + e(x)),\] where $s(x)\in R_q$ is the secret and $e(x) \in R_q$ is the error with coefficients subject to the error distribution. 

Let $\textbf{a},\textbf{s} \mbox{ and } \textbf{e} \in \ZZ_q^n$ be the coefficient vectors of $a(x),s(x)$ and $e(x)$. Then the coefficient vector $\textbf{b}$ of $b(x)$ can be obtained via the formula 
\[\textbf{b} = \textbf{A}_{a(x)}^{\text{circ}}\cdot \textbf{s} + \textbf{e},\] 
here $\textbf{A}_{a(x)}^{\text{circ}}$ represents the $n\times n$ generalized circulant matrix of $a(x)$. Precisely, let $a(x) = a_0 + a_1x + \ldots + a_{n-2}x^{n-2} + a_{n-1}x^{n-1}$, then $\textbf{a} = (a_0,a_1,\ldots,a_{n-2},a_{n-1})$ and  \[\textbf{A}_{a(x)}^{\text{circ}} = \begin{bmatrix}
    a_0       & -a_{n-1} & -a_{n-2} & \dots & -a_1 \\
    a_1       & a_0 & -a_{n-1} & \dots & -a_2 \\
    a_2       & a_1 & a_0 & \dots & -a_3\\
    \vdots       & \vdots & \vdots & \ddots & \vdots    \\
    a_{n-1}     & a_{n-2} & a_{n-3} & \dots & a_0
\end{bmatrix}.\]
Therefore, one RLWE sample gives rise to $n$ LWE instances by taking the rows of $\textbf{A}_{a(x)}^{\text{circ}}$ and the corresponding entries in ${\bf b}$.

\subsection{Search to Decision Reduction for Binary Secrets (\S\ref{sec:lattice})}\label{subsec:search_decision} We give a proof of the search binary-LWE to decisional binary-LWE reduction. This is a simple adaption of the reduction in \cite{Reg05} to the binary secrets case. We call an algorithm a $(T,\gamma)$-distinguisher for two probability distributions $\mathcal{D}_0,\mathcal{D}_1$ if it runs in time $T$ and has a distinguishing advantage $\gamma$. We use $\lwe_{n,m,q,\chi}$ to denote the LWE problem which has secret dimension $n$, $m$ LWE instances, modulus $q$ and the secret distribution $\chi$.

\begin{theorem}
If there is a $(T,\gamma)$-distinguisher for decisional binary-$\lwe_{n,m,q,\chi}$, then there is a $T' = \widetilde{O}(Tn/\gamma^2)$-time algorithm that solves search binary-$\lwe_{n,m',q,\chi}$ with probability $1-o(1)$, where $m' = \widetilde{O}(m/\gamma^2)$.
\end{theorem}

\begin{proof}
Let $\textbf{s} = (s_1,\ldots,s_n)$ with $s_i \in \{0,1\}$. We demonstrate the strategy of recovering $s_1$, and the rest of the secret coordinates can be recovered in the same way. Let $m' = \widetilde{O}(1/\gamma^2)m$, given an LWE sample $(\textbf{A},\,\textbf{b})$ where $\textbf{A}\in \ZZ_q^{m'\times n},\textbf{b} \in \ZZ_q^{m'}$, we compute a pair $\bf (A',\,b')$ as follows:
\[\textbf{A}' = \textbf{A} +' \textbf{c},\,\,\, \textbf{b}' = \textbf{b}.\]

Here $\textbf{c} \in \ZZ_q^{m'}$ is sampled uniformly and the symbol $``+' "$ means that we are adding $\textbf{c}$ to the first column of $\textbf{A}$. One verifies by the definition of LWE that if $s_1 = 0$, then the pair $(\textbf{A}',\, \textbf{b}')$ would be LWE samples with the same error distribution. Otherwise, the pair $(\textbf{A}',\, \textbf{b}')$ would be uniformly random in $\ZZ_q^{m'\times n}\times \ZZ_q^{m'}$. We then feed the pair $(\textbf{A}',\, \textbf{b}')$ to the $(T,\gamma)$-distinguisher for $\lwe_{n,m,q,\chi}$, and we need to running the distinguisher $m'/m = \widetilde{O}(1/\gamma^2)$ times given the number of instances. Since the advantage of this distinguisher is $\gamma$ with $m$ LWE instances, and we are feeding it $m'$ LWE instances, it follows from the Chernoff bound that if the majority of the outputs are ``LWE", then the pair $\bf (A', b')$ is an LWE sample and therefore $s_1 = 0$. If not, $s_1 = 1$. Guessing one coordinate requires running the distinguisher $ \widetilde{O}(1/\gamma^2)$ times, therefore, this search to reduction algorithm takes time $T' = \widetilde{O}(Tn/\gamma^2)$. Note that we can use the same $m'$ LWE instances for each coordinate, therefore it requires $m' = \widetilde{O}(m/\gamma^2)$ samples to recover all the secret coordinates.
\end{proof}

\section{Additional Modular Arithmetic Results (\S\ref{sec:results_1D})}
\label{sec:appx_addl_1D}

\begin{center}
\begin{table}[th]

\centering
\begin{tabular}{ll|ll}
\toprule
\textbf{$\ceil{\log_2(q)}$} & \textbf{$q$} & \textbf{$\ceil{\log_2(q)}$} & \textbf{$q$} \\ \midrule
5 & 19, 29 & 18 & 147647, 222553 \\
6 & 37, 59 & 19 & 397921, 305423 \\
7 & 67, 113 & 20 & 842779, 682289 \\
8 & 251, 173 & 21 & 1489513, 1152667 \\
9 & 367, 443 & 22 & 3578353, 2772311 \\
10 & 967, 683 & 23 & 6139999, 5140357 \\
11 & 1471, 1949 & 24 & 13609319, 14376667 \\
12 & 3217, 2221 & 25 & 31992319, 28766623 \\
13 & 6421, 4297 & 26 & 41223389, 38589427 \\
14 & 11197, 12197 & 27 & 94056013, 115406527 \\
15 & 20663, 24659 & 28 & 179067461, 155321527 \\
16 & 42899, 54647 & 29 & 274887787, 504470789 \\
17 & 130769, 115301 & 30 & 642234707, 845813581 \\ \bottomrule
\end{tabular}

\vspace{0.1cm}
\caption{\small $q$ values used in our experiments} 
\label{tab:Qs}
\end{table}
\end{center}

Here, we provide additional information on our single and multidimensional modular arithmetic experiments from \S\ref{sec:1D_method}. Before presenting experimental results, we first highlight two useful tables. Table~\ref{tab:Qs} shows the $q$ values used in our integer and multi-dimension modular arithmetic problems. Table~\ref{tab:app_1Dspeed} is an expanded version of Table~\ref{tab:1Dspeed} in the main paper body. It shows how the $\log_2$ samples required for success changes with the base representation for the input/output, but includes additional values of base $B$ (secret is fixed at $728)$.

\begin{table}[h]
\centering
\begin{tabular}{ccccccccccccccc}\toprule
\multicolumn{1}{l}{\multirow{2}{*}{$\boldsymbol{\ceil{\log_2(q)}}$}} & \multicolumn{13}{c}{\textbf{Base}} \\ \multicolumn{1}{l}{} & 2 & 3 & 4 & 5 & 7 & 17 & 24 & 27 & 30 & 31 & 63 & \bf{81} & 128 \\ \midrule
15 & 23 & 21 & 21 & 23 & 22 & 20 & 20 & 23 & 22 & 21 & 21 & \textbf{20} & 20 \\
16 & 24 & 22 & 23 & 22 & 22 & 23 & 22 & 22 & 22 & 23 & 22 & \textbf{22} & 21 \\
17 & - & 23 & 24 & 25 & 22 & 26 & 23 & 24 & 22 & 24 & 23 & \textbf{22} & 22 \\
18 & - & 23 & 23 & 25 & 23 & - & 23 & 24 & 25 & - & 23 & \textbf{22} & 22 \\
19 & - & 23 & - & - & 25 & 23 & 25 & 24 & - & - & 25 & \textbf{25} & 24 \\
20 & - & - & - & - & - & 24 & 25 & 24 & 26 & - & - & \textbf{24} & 25\\
21 & - & 24 & - & - & 25 & - & - & - & - & - & - & - & 25 \\
22 & - & - & - & - & - & - & - & 25 & - & 26 & - & - & 25 \\
23 & - & - & - & - & - & - & - & - & - & - & 25 & - & - \\
24 & - & - & - & - & - & - & - & - & - & - & - & - & - \\ \bottomrule
\end{tabular}
\vspace{0.1cm}
\caption{\small Base-2 logarithm of the number of examples needed to reach $95\%$ accuracy, for different values of $\ceil{\log_2(q)}$ and bases.}
\label{tab:app_1Dspeed}
\end{table}

\para{Base vs. Secret.} We empirically observe that the base $B$ used for integer representation in our experiments may provide side-channel information about the secret $s$ in the 1D case. For example, in Table~\ref{tab:1D_baseSecret}, when the secret value is $729$, bases $3$, $9$, $27$, $729$ and $3332$ all enable solutions with much higher $q$ (8 times higher than the next highest result). Nearly all these are powers of 3\footnote{And $3332$ can easily be written out as a sum of powers of $3$, e.g. $3332=3^8-3^7-3^6-3^5-3^4+3^2+3-1$.} as is the secret $729=3^6$. In the table, one can see that these same bases provide similar (though not as significant) ``boosts" in $q$ for secrets on either side of $729$ (e.g. $728, 730$), as well as for $720 = 3^6-3^2$. Based on these results, we speculate that when training on ($a,b$) pairs with an unknown secret $s$, testing on different bases and observing model performance may allow some insight into $s$'s prime factors. More theoretical and empirical work is needed to verify this connection.

\begin{table}[h]
\centering
\begin{tabular}{cccccccccccc}\toprule
\multirow{2}{*}{\textbf{Base}} & \multicolumn{11}{c}{\textbf{Secret value}} \\ \cmidrule{2-12}
 & 720 & 721 & 722 & 723 & 724 & 725 & 726 & 727 & 728 & 729 & 730 \\ \midrule
$2$ & - & 18 & 16 & - & - & - & 16 & 16 & - & 16 & - \\
$3$ & 19 & 16 & 18 & 16 & 18 & 18 & 19 & 18 & 21 & \textbf{24} & 20 \\
$4$ & 18 & 18 & 18 & 18 & 18 & - & - & - & 18 &  & 18 \\
$5$ & 18 & 17 & 16 & 16 & 16 & 18 & 18 & 17 & 16 & 19 & 16\\
$7$ & - & 18 & - & 18 & - & 20 & - & - & 19 & 18 & - \\
$9$ & \textbf{23} & 18 & 18 & 18 & 18 & 18 & 18 & 21 & 18 & \textbf{24} & \textbf{23}\\
$11$ & 20 & 20 & - & 19 & 21 & 21 & 21 & 20 & - & - & 19 \\
$17$ & 18 & 18 & - & 19 & 19 & 18 & 18 &  & 20 & 20 & - \\
$27$ & \textbf{23} & - & \textbf{23} & 18 & 18 & 18 & 21 & 22 & 21 & \textbf{24} & 22 \\
$28$ & - & 20 & 18 & - & 20 & 18 & - & 19 & \textbf{23} & 18 & - \\
$49$ & 18 & 22 & 18 & 19 & 21 & 18 & - & 18 & 19 & 18 & 18  \\
$63$ & 20 & 21 & 18 & 20 & 19 & 19 & 18 & 19 & 18 & - & -  \\
$128$ & 20 & - & 18 & 22 & 20 & - & 19 & 18 & 19 & 19 & 19 \\
$729$ & 18 & 20 & 19 & 18 & 19 & 21 & 19 & 19 & 18 & \textbf{25} & 18 \\
$3332$ & 22 & 22 & 22 & \textbf{23} & 22 & 22 & 21 & 22 & \textbf{23} & \textbf{23} & 22 \\ \bottomrule
\end{tabular}
\vspace{0.1cm}
\caption{\small Relationship between base and secret. Numbers in table represent the highest $\log_2(q)$ value achieved for a particular base/secret combo. Values of $\log_2(q)>=23$, indicating high performance, are {\bf bold}.}
\label{tab:1D_baseSecret}

\end{table}

\vspace{-0.35cm}

\para{Ablation over transformer hyper-parameters.} We provide additional experiments on model architecture, specifically examining the effect of model layers, optimizer, embedding dimension and batch size on integer modular inversion performance. Tables~\ref{tab:appx_1D_num_layers}-\ref{tab:appx_1D_batch_size} show ablation studies for the 1D modular arithmetic task, where entries are of the form (best $\log_2(q)$/$\log_2(samples)$), e.g. the highest modulus achieved and the number of training samples needed to achieve this. The best results, meaning the highest $q$ with the lowest $\log_2(samples)$, are in {\bf bold}. For all experiments, we use the base architecture of 2 encoder/decoder layers, 512 encoder/decoder embedding dimension, and 8/8 attention heads (as in Section~\ref{sec:1D_method}) and note what architecture element changes in the table heading. 

We find that shallow transformers (e.g 2 layers, see Table~\ref{tab:appx_1D_num_layers}) work best, allowing to solve problems with a much higher $q$ especially when the base $B$ is large. The AdamCosine optimizer (Table~\ref{tab:appx_1D_optimizer}) usually works best, but requires smaller batch sizes for success with large bases. For small bases, a smaller embedding dimension of 128 performs better (Table~\ref{tab:appx_1D_emb_dim}), but increasing base size and dimension simultaneously yield good performance. Results on batch size (Table~\ref{tab:appx_1D_batch_size}) do not show a strong trend.

\begin{table}[h]
\begin{center}
\begin{minipage}[b]{0.32\linewidth}

\centering
    \resizebox{1.0\textwidth}{!}{
    \begin{tabular}{cccc}\toprule
    \multirow{2}{*}{\textbf{Base}} & \multicolumn{3}{l}{\textbf{\# Transformer Layers}} \\ \cmidrule{2-4}
     & 2 & 4 & 6 \\ \midrule
    27 & 19/24 & 18/27 & 20/25 \\
    63 & 18/25 & 16/25 & 15/22 \\
    3332 & \textbf{23/26} & 23/-- & 18/22 \\ \bottomrule
    \end{tabular}
    }
        \vspace{0.1cm}
    \caption{\small 1D case: Ablation over number of transformer layers.}
    \label{tab:appx_1D_num_layers}

\end{minipage}\hfill%
\begin{minipage}[b]{0.66\linewidth}
\centering
    \resizebox{1.1\textwidth}{!}{
    \begin{tabular}{cccc}\toprule
    \multirow{2}{*}{\textbf{Base}} & \multicolumn{3}{c}{\textbf{Optimizer}} \\ \cmidrule{2-4}
     & Adam (0, $5e^{-5}$) & Adam ($3000, 5e^{-5}$) & AdamCosine ($3000, 1e^{-5}$) \\ \midrule
    27 & 18/26 & 19/24 & \textbf{22/27} \\
    3332 & \textbf{ 23/26} & 23/27 & 22/26 \\
    3332* & 23/26 & 23/26 & \textbf{ 23/25} \\ \bottomrule
    \end{tabular}
    }
    \vspace{0.1cm}
    \caption{\small 1D case: Ablation over optimizers. Parenthetical denotes (\# warmup steps, learning rate); * = batch size 128.}
    \label{tab:appx_1D_optimizer}

\end{minipage}\qquad\vspace{0.1cm}
\begin{minipage}[h]{0.48\linewidth}
\centering
    \begin{tabular}{ccccc}\toprule
    \multirow{2}{*}{\textbf{Base}} & \multicolumn{4}{c}{\textbf{Embedding Dimension}} \\ \cmidrule{2-5}
     & 512 & 256 & 128 & 64 \\ \midrule
    3 & 21/25 & 21/24 & \textbf{22/26} & 19/- \\
    27 & 23/26 & 23/26 & \textbf{23/25} & 19/26 \\
    63 & \textbf{ 23/27} & 18/24 & 19/27 & 18/26 \\
    3332 & \textbf{23/25} & 23/26 & 23/26 & 23/27 \\ \bottomrule
    \end{tabular}
        \vspace{0.1cm}
        \caption{\small 1D case: Ablation over embedding.}
    \label{tab:appx_1D_emb_dim}
    \vspace{-0.25cm}
\end{minipage}\hfill
\begin{minipage}[h]{0.48\linewidth}
\centering
    \begin{tabular}{llllll} \toprule
    \multirow{2}{*}{\textbf{Base}} & \multicolumn{5}{c}{\textbf{Batch size}} \\ \cmidrule{2-6}
     & 64 & 96 & 128 & 192 & 256 \\ \midrule
    3 & 21/26 & 21/25 & 21/26 & 22/26 & \textbf{23/26} \\
    27 & 21/25 & \textbf{24/27} & 22/27 & 23/26 & 24/28 \\
    63 & - & 20/27 & \textbf{23/25} & - & 23/26 \\
    3332 & 23/26 & 23/26 & \textbf{23/25} & \textbf{23/25} & 23/26 \\ \bottomrule
    \end{tabular}
        \vspace{0.1cm}
       \caption{\small 1D case: Ablation over training batch size.}
    \label{tab:appx_1D_batch_size}
    \vspace{-0.25cm}
 
\end{minipage}
\end{center}
\end{table}

\vspace{-0.25cm}

\section{Additional information on SALSA Secret Recovery (\S\ref{subsec:secret_guess})}
\label{sec:appx_secret_guess}

\subsection{Direct Secret Recovery}

\para{Recovering Secrets from Predictions.} In the direct secret recovery phase, the model predicts, for each value of K, $n$ sequences representing integers in base $B$ (one for each special ${\bf a}$ input). They are decoded as $n$ integers, and concatenated as a vector $\tilde{\bf s}$. The function $binarize$, on line $7$ of Algorithm~\ref{alg:direct_secret}, then predicts the binary secret from $\tilde{\bf s}$. $binarize$ outputs six predictions, using three methods: mean, softmax-mean, and mode comparison. 
\begin{packed_itemize}
\item The mean comparison method takes the mean of the coordinates of $\tilde{\bf s}$ and computes \textit{two} potential secrets: $f_{01}(\tilde{\bf s})$ where all coordinates above the mean as set to $0$ and all below the mean to $1$, and $f_{10}(\tilde{\bf s})$ where all coordinates above the mean are set to $1$ and all below the mean to $0$.
\item In the softmax-mean comparison method, we apply a softmax function to $\tilde{\bf s}$ before using the mean comparison method, to obtain two secret predictions. 
\item The mode comparison method uses the mode (the most common value) of $\tilde{\bf s}$ instead of the mean. 
\end{packed_itemize}

Altogether, these binarization methods produce six secret guesses. In our SALSA evaluation, all of these are compared against the true secret $\bf s$, and the number of matching bits is reported. If $\tilde{\bf s}$ fully matches $\bf s$, model training is stopped. When $\bf s$ is not available for comparison, the methods in \S\ref{sec:check_secrets} can be used to verify $\tilde{\bf s}$'s correctness.

\para{K values.} At the end of each epoch, we use $10$ $K$ values for direct secret guessing, $5$ of which are fixed and $5$ of which are randomly generated. The fixed $K$ values are $K = [239145, 42899, q-1, 3q+7, 42900]$, while the random $K$ values are chosen from the range $(q, 10q)$.

\subsection{Distinguisher-Based Secret Recovery}

Here, we provide more details on the parameters and subroutines used in Algorithm~\ref{alg:distinguisher}. 
\begin{packed_itemize}
\item $\tau$: This parameter sets the bound on $q$ that will be used for the distinguisher computation. In our experiments, we set $\tau=0.1$.
\item $acc_{\tau}$: This denotes the distinguisher advantage. Let $acc_{\tau}$ denote the proportion of model predictions which fall within the chosen tolerance (e.g. accuracy within tolerance as described in \S\ref{subsec:model_train}), then $advantage$ = $acc_{\tau} - 2*\tau$.
\item $LWESamples(t,n,q)$ is a subroutine that returns LWE samples $({\bf A, b})\in \ZZ_q^{n\times t}\times \ZZ_q^t$, note that now columns of ${\bf A}$ corresponds to LWE instances.
\end{packed_itemize}

\subsection{Secret Recovery in Practice}
\label{sec:sec_histogram}

\begin{figure}[h]
    \centering
    \includegraphics[width=0.4\textwidth]{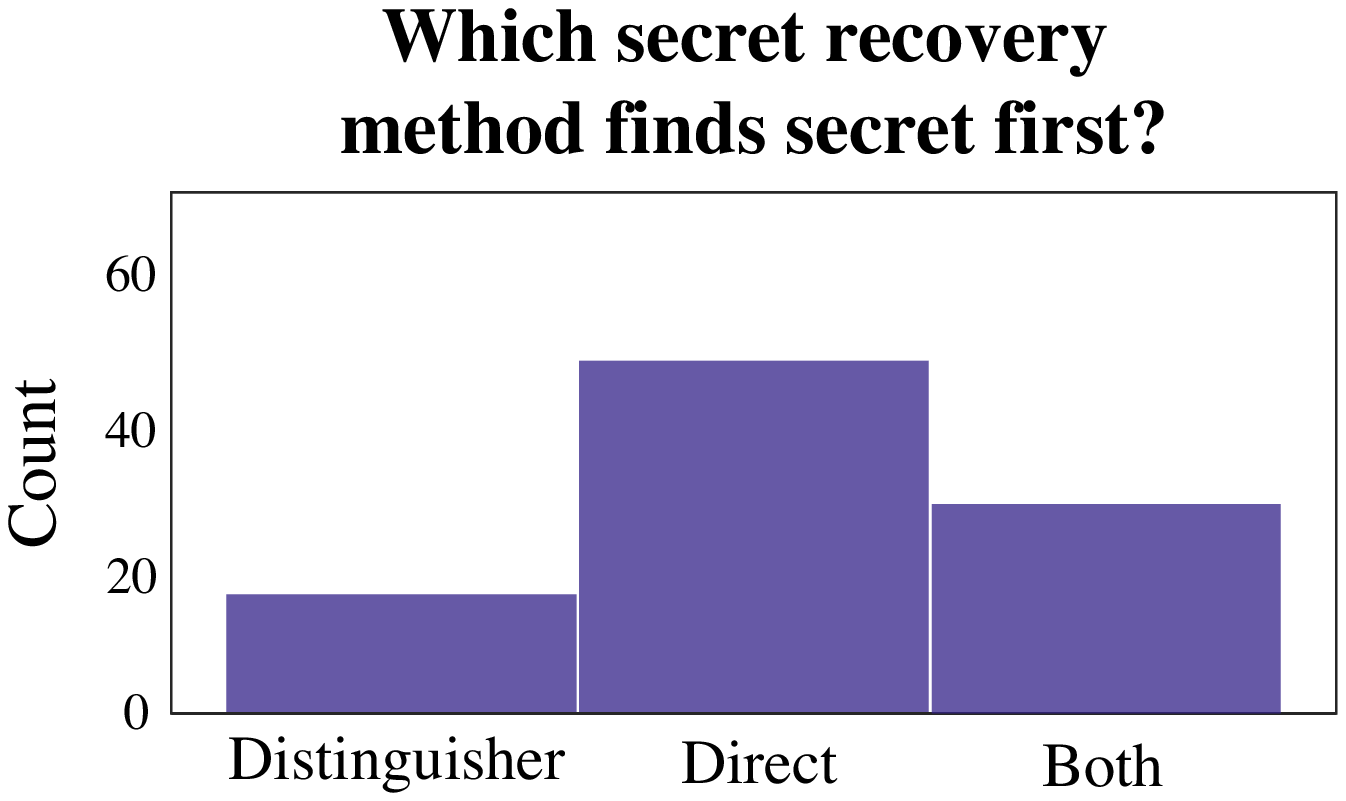}
    \caption{\small Frequency counts of which secret recovery method succeeds first for $90$ successful SALSA runs.}
    \label{fig:secret_histogram}
\end{figure}

Empirically, we observe that our direct secret guessing recovers the secret more quickly than the distinguisher-based method. Figure~\ref{fig:secret_histogram} plots the number of times each technique succeeded over $120$ SALSA runs with varying $n$ and $d$. The direct secret guessing method succeeds in $80\%$ of cases. Occasionally ($20\%$), the distinguisher finds the secret first. Both methods simultaneously succeed in $30\%$ of the cases.


\section{SALSA Architecture experiments}
\label{sec:appx_2D_arch}

Several key architecture choices determine SALSA's ability to recover secrets with higher $n$ and $d$, namely the encoder and decoder dimension as well as the number of attention heads. Other architecture choices determine the time to solution but not the complexity of problems SALSA could solve. For example, universal transformers (UT) are more sample efficient than regular transformers. Using gated loops in the UT with more loops on the decoder than the encoder reduced both model training time and the number of samples needed. Here, we present ablation results for all these architectural choices.

\vspace{-0.25cm}

\begin{figure*}[h]
\begin{minipage}[c]{0.30\textwidth}
    \vspace{0.35cm}
    \centering
    \resizebox{1\textwidth}{!}{
    \begin{tabular}{cccc}\toprule
    \multirow{2}{*}{\textbf{\begin{tabular}[c]{@{}c@{}}Encoder\\ Loops/Layers\end{tabular}}} & \multicolumn{3}{c}{\textbf{\begin{tabular}[c]{@{}c@{}}Decoder\\ Loops/Layers\end{tabular}}} \\
     & 2 & 4 & 8\\ \midrule
    2 & 1.2 & 4.7 & 0.8 \\
    4 & 0.7 & 0.4 & 0.6 \\
    8 & 0.1 & 0.1 & 0.1 \\ \bottomrule
    \end{tabular}
    }
    \captionof{table}{\small {\bf Transformers vs UTs.} Ratio of training samples required for success for UTs with X/X encoder/decoder loops vs regular transformers with X/X encoder/decoder layers.}
    \label{tab:2D_RT_UT}
\end{minipage}\hfill
\begin{minipage}[c]{0.32\textwidth}
     \centering
     \small
     \begin{tabular}{cccc}\toprule
     
    \multirow{2}{*}{\textbf{\begin{tabular}[c]{@{}c@{}}Encoder\\ Loops\end{tabular}}} & \multicolumn{3}{c}{\textbf{ Decoder Loops}} \\ 
     & \textbf{2} & \textbf{4} & \textbf{8} \\ \midrule
    2 & 1.0 & 1.3 & 0.3 \\
    4 & 0.3 & 0.3 & 0.3 \\
    8 & 0.1 & 0.1 & 0.1 \\ \bottomrule
    \end{tabular}
    \captionof{table}{\small \textbf{Gated vs Ungated UTs.} Ratio of training samples required for success for gated UTs with X/X encoder/decoder loops vs ungated UTs with same loop numbers.}
    \label{fig:2D_gated_ungated}
    \end{minipage}\hfill
\begin{minipage}[c]{0.32\textwidth}
     \centering
     \small
     \begin{tabular}{cccc}\toprule
     
    \multirow{2}{*}{\textbf{\begin{tabular}[c]{@{}c@{}}Encoder\\Loops\end{tabular}}} & \multicolumn{3}{c}{\textbf{Decoder Loops}} \\
     & \textbf{2} & \textbf{4} & \textbf{8} \\ \midrule
    \multicolumn{1}{c} 2 & 23.5 & 25.4 & 23.4 \\
    \multicolumn{1}{c} 4 & 23.3 & 24.2 & 24.4 \\
    \multicolumn{1}{c} 8 & 23.1 & 22.3 &  22.5 \\ \bottomrule
    \end{tabular}
     \captionof{table}{\small {\bf Loops.} Average $\log_2$ of training samples required for $N=50$, $h=3$, $q=251$, $base_{in}/base_{out}=81$ as loops vary.}
    \label{fig:appx_compare_loops}
\end{minipage}
\end{figure*}
\vspace{-0.25cm}

\para{Universal Transformers vs. Regular Transformers.} First, we compare universal transformers (UT) with ``regular'' transformers. We run dueling experiments on medium size problems ($N=50$, $d=0.06$, $q=251$, $B_{in}/B_{out}=81$), comparing gated UT and 2 to 8 loops to regular transformers with is many layers as the UT has loops. For each pair of experiment, we measure the ratio of the number of training samples needed to recover the secret. As Table~\ref{tab:2D_RT_UT} shows, universal transformers prove more sample efficient as model size increases. We use them exclusively.

\para{Gated vs Ungated UTs.} To understand the effect of gating on sample efficiency, we run two experiments with medium-size problems ($N=50$, $d=0.06$, $q=251$, $B_{in}/B_{out}=81$), with gated and ungated universal transformers with 2 to 8 loops in the encoder  and decoder. Table~\ref{fig:2D_gated_ungated} proves that gated UTs are much more sample-efficient.

\para{Number of Loops.} Table~\ref{fig:appx_compare_loops} reports the logarithm of the average number of samples needed to recover the secret for $N=50$ and $h=3$, for various numbers of loops in the encoder and decoder (dimensions=1025/512, heads=16/4). 8/4 and 8/8 loops prove more efficient, but because training time increases steeply with the number of loops in the encoder (which processes longer sequences), we kept $2$ encoder and $8$ decoder loops in our experiments.

\para{Encoder/Decoder Dimension.} Table~\ref{tab:layers_loops} in \S\ref{subsec:salsa_results} show that larger encoders and smaller decoders improve sample efficiency. Here, we explore how encoder and decoder dimension impact the size of the secrets (in terms of dimension $n$ and hamming weight) that SALSA can recover. Our results, shown in Tables~\ref{tab:appx_enc_dim} and~\ref{tab:appx_dec_dim}, follow the same pattern as before: large encoders and small decoder allow for the recovery of secrets of higher dimension $n$. Furthermore, for $n=30$, larger encoders and/or smaller decoders allow for the recovery of secrets with hamming weight $4$. 

\begin{table}[h]
\centering
\begin{tabular}{ccccc?cccc}\toprule
\multirow{2}{*}{$n$} & \multicolumn{4}{c?}{\begin{tabular}[c?]{@{}c@{}}\textbf{Encoder Dimension}\\ (hamming=3)\end{tabular}} & \multicolumn{4}{c}{\begin{tabular}[c]{@{}c@{}}\textbf{Encoder Dimension}\\ (hamming=4)\end{tabular}} \\ 
 & 512 & 1024 & 2048 & 3040 & 512 & 1024 & 2048 & 3040 \\ \midrule
30 &\cellcolor{lightgreen!50} 1.0 & \cellcolor{lightgreen!50} 1.0 &\cellcolor{lightgreen!50} 1.0 &\cellcolor{lightgreen!50} 1.0 & \cellcolor{lred!80} 0.87 &\cellcolor{lightgreen!50} 1.0 & \cellcolor{lightgreen!50} 1.0 & \cellcolor{lightgreen!50} 1.0 \\

50 & \cellcolor{lightgreen!50}1.0 &\cellcolor{lightgreen!50} 1.0 & \cellcolor{lightgreen!50} 1.0 & \cellcolor{lightgreen!50} 1.0 & \cellcolor{lyellow!80}0.94 &\cellcolor{lyellow!80} 0.94 & \cellcolor{lyellow!80}0.94 & \cellcolor{lyellow!80}0.94 \\

70 & \cellcolor{lyellow!80} 0.97 &\cellcolor{lightgreen!50} 1.0 & \cellcolor{lightgreen!50}1.0 & \cellcolor{lightgreen!50}1.0 &\cellcolor{lyellow!80} 0.96 & \cellcolor{lyellow!80}0.97 & \cellcolor{lyellow!80}0.94 & 0\cellcolor{lyellow!80}.96  \\

90 & \cellcolor{lred!80}0.97 & \cellcolor{lyellow!80}0.98 &\cellcolor{lightgreen!50} 1.0 & \cellcolor{lightgreen!50}1.0 & \cellcolor{lred!80} 0.96 & \cellcolor{lyellow!80}0.96 & \cellcolor{lyellow!80}0.97 & \cellcolor{lyellow!80}0.97  \\ \bottomrule
\end{tabular}

\vspace{0.1cm}
\caption{\small{\bf Ablation over Encoder Dimension.} Proportion of secret bits recovered for varying $n$ and encoder dimension. For all experiments, we fix decoder dimension to be 512, 2/2 layers, 2/8 loops. {\color{lightgreen} Green} means secret was guessed, {\color{lyellow} yellow} means all $1$s, but not all $0$s, were guessed, and {\color{lred} red} means SALSA failed.}
\label{tab:appx_enc_dim}

\end{table}

\begin{table}[h]
\centering

\begin{tabular}{ccccc?cccc}\toprule
\multirow{2}{*}{$n$} & \multicolumn{4}{c?}{\begin{tabular}[c]{@{}c@{}}\textbf{Decoder Dimension}\\ (hamming=3)\end{tabular}} & \multicolumn{4}{c}{\begin{tabular}[c]{@{}c@{}}\textbf{Decoder Dimension}\\ (hamming=4)\end{tabular}} \\ 
 &  256 & 768 & 1024 & 1536 &  256 & 768 & 1024 & 1536 \\ \midrule
30 & \cellcolor{lightgreen!50}1.0 & \cellcolor{lightgreen!50}1.0 & \cellcolor{lightgreen!50}1.0 & \cellcolor{lightgreen!50}1.0 &  \cellcolor{lightgreen!50}1.0 & \cellcolor{lightgreen!50}1.0 & \cellcolor{lyellow!80}0.90 &\cellcolor{lyellow!80} 0.87 \\

50 & \cellcolor{lightgreen!50}1.0 & \cellcolor{lightgreen!50}1.0 & \cellcolor{lightgreen!50}1.0 & \cellcolor{lyellow!80}0.94 & \cellcolor{lyellow!80}0.94 & \cellcolor{lred!80}0.92 &\cellcolor{lyellow!80} 0.92 & \cellcolor{lyellow!80}0.92 \\

70 & \cellcolor{lightgreen!50}1.0 & \cellcolor{lightgreen!50}1.0 & \cellcolor{lightgreen!50}1.0 & \cellcolor{lyellow!80}0.96 &\cellcolor{lyellow!80} 0.96 & \cellcolor{lred!80}0.94 & \cellcolor{lyellow!80}0.94 & \cellcolor{lyellow!80}0.94  \\

90 & \cellcolor{lightgreen!50}1.0 & \cellcolor{lyellow!80}0.97 &\cellcolor{lyellow!80} 0.97 &\cellcolor{lyellow!80} 0.97 & \cellcolor{lyellow!80} 0.97 & \cellcolor{lred!80}0.96 & \cellcolor{lyellow!80}0.97 & \cellcolor{lred!80} -  \\ \bottomrule
\end{tabular}

\vspace{0.1cm}
\caption{\small {\bf Ablation over Decoder Dimension.} Proportion of secret bits recovered for varying $n$ and encoder dimension. For all experiments, we fix encoder dimension to be 1024, 2/2 layers, 2/8 loops. {\color{lightgreen} Green} means secret was guessed, {\color{lyellow} yellow} means all $1$s, but not all $0$s, were guessed, and {\color{lred} red} means SALSA failed.}
\label{tab:appx_dec_dim}

\end{table}

\para{Attention Heads.} In these experiments, we train universal transformers with 2 encoder/decoder layers, 1024/512 embedding dimension, 2/8 encoder/decoder loops, with different number of attentions heads, on problems of different dimensions $n$, and measure SALSA secret recovery rate. Increasing the number of attention heads in the encoder, and reducing it to $4$ in the decoder allows SALSA to recover secrets for $n=90$ (Table~\ref{tab:heads_success}), although it slightly increases the number of samples needed for recovery (Table~\ref{tab:heads_samples}). Increasing the number of decoder heads increases the number of samples needed but does not provide the same scale-up for $n=90$.

\begin{table}[h]
\centering
\begin{tabular}{ccccccccc}\toprule
\multirow{2}{*}n & \multicolumn{8}{c}{\textbf{Encoder/Decoder Heads}} \\
 & 8/8 & 16/4 & 16/8 & 16/16 & 32/4 & 32/8 & 32/16 & 32/32 \\ \midrule
30 & \cellcolor{lightgreen!50} 1.0 & \cellcolor{lightgreen!50} 1.0 & \cellcolor{lightgreen!50} 1.0 &\cellcolor{lightgreen!50} 1.0 & \cellcolor{lightgreen!50} 1.0 & \cellcolor{lightgreen!50} 1.0 & \cellcolor{lightgreen!50} 1.0 & \cellcolor{lightgreen!50} 1.0  \\

50 & \cellcolor{lightgreen!50} 1.0 & \cellcolor{lightgreen!50} 1.0 & \cellcolor{lightgreen!50} 1.0 & \cellcolor{lightgreen!50} 1.0 & \cellcolor{lightgreen!50} 1.0 & \cellcolor{lightgreen!50} 1.0 & \cellcolor{lightgreen!50} 1.0 & \cellcolor{lightgreen!50} 1.0  \\

70 & \cellcolor{lightgreen!50} 1.0 & \cellcolor{lightgreen!50} 1.0 & \cellcolor{lightgreen!50} 1.0 & \cellcolor{lightgreen!50} 1.0 & \cellcolor{lightgreen!50} 1.0 & \cellcolor{lightgreen!50} 1.0 & \cellcolor{lightgreen!50} 1.0 & \cellcolor{lightgreen!50} 1.0  \\

90 & \cellcolor{lyellow!80} 0.97 & \cellcolor{lightgreen!50} 1.0 &\cellcolor{lyellow!80}  0.97 & \cellcolor{lyellow!80} 0.99 & \cellcolor{lightgreen!50} 1.0 & \cellcolor{lyellow!80} 0.98 & \cellcolor{lyellow!80} 0.98 & \cellcolor{lyellow!80} 0.97 \\ \bottomrule 
\end{tabular}
\vspace{0.1cm}
\caption{\small {\bf  Attention Heads: Effect on secret recovery.} Table of success for varying $n$ with hamming 3 for encoder/decoder head combinations. {\color{lightgreen} Green} means secret was guessed, {\color{lyellow} yellow} means all $1$s, but not all $0$s, were guessed.}
\label{tab:heads_success}

\end{table}

\begin{table}[h]
\centering

\begin{tabular}{c|c|c}\toprule
\multicolumn{3}{c}{\begin{tabular}[c]{@{}c@{}} \textbf{Attention Heads}\\ (1024/512, X/X, 2/8)\end{tabular}} \\ \midrule
8/8 & 16/4,8,16 & 32/4,8,16,32 \\ \midrule
22.4 & 22.8, 22.9, 23.2 & 23.0, 23.1, 23.7, 24.7 \\ \bottomrule
\end{tabular}
\vspace{0.1cm}
\caption{\small {\bf Attention Heads: Effect on $\log_2$ samples.} We test the effect of attention heads and report the $\log_2$ samples required to recover the secret in each setting. Experiments are run with $n=50$, hamming $3$.}
\label{tab:heads_samples}
\end{table}

\end{document}